\documentclass[prd,tightenlines,nofootinbib,superscriptaddress]{revtex4}

\usepackage{amsfonts,amssymb,amsthm,bbm}

\usepackage{amsmath}

\usepackage{hyperref}

\usepackage{subcaption}
\usepackage{color,psfrag}
\usepackage[dvips]{graphicx}
\usepackage{tikz}
\usetikzlibrary{calc}
\usetikzlibrary{decorations.pathmorphing}
\usetikzlibrary{decorations.markings}

\newcommand{\C}{{\mathbb C}}
\newcommand{\N}{{\mathbb N}}
\newcommand{\R}{{\mathbb R}}
\newcommand{\Z}{{\mathbb Z}}

\newcommand{\cA}{{\mathcal A}}

\newcommand{\cG}{{\mathcal G}}

\newcommand{\cH}{{\mathcal H}}

\newcommand{\cN}{{\mathcal N}}

\newcommand{\cT}{{\mathcal T}}
\newcommand{\cV}{{\mathcal V}}
\newcommand{\cD}{{\mathcal D}}
\newcommand{\cC}{{\mathcal C}}

\newcommand{\cS}{{\mathcal S}}

\newcommand{\SU}{\mathrm{SU}}

\newcommand{\SL}{\mathrm{SL}}
\newcommand{\SO}{\mathrm{SO}}

\newcommand{\U}{\mathrm{U}}

\newcommand{\be}{\begin{equation}}
\newcommand{\ee}{\end{equation}}
\newcommand{\beq}{\begin{eqnarray}}
\newcommand{\eeq}{\end{eqnarray}}
\newcommand{\bes}{\begin{eqnarray}}
\newcommand{\ees}{\end{eqnarray}}

\newcommand{\mat} [2] {\left ( \begin{array}{#1}#2\end{array} \right ) }

\renewcommand{\u}{{\mathfrak{u}}}
\newcommand{\su}{{\mathfrak{su}}}

\newcommand{\so}{{\mathfrak{so}}}

\newcommand{\la}{\langle}
\newcommand{\ra}{\rangle}

\newcommand{\tr}{{\mathrm{Tr}}}
\newcommand{\f}{\frac}

\def\nn{\nonumber}
\def\pp{\partial}

\def\vphi{\varphi}
\def\eps{\epsilon}
\def\om{\omega}

\newcommand{\id}{\mathbb{I}}

\def\act{\triangleright}

\def\vsigma{\vec{\sigma}}

\def\act{\,\triangleright\,}
\def\bz{\bar{z}}

\def\dd{\mathrm{d}}
\def\rd{\mathrm{d}}
\def\vX{\vec{X}}

\def\balpha{\bar{\alpha}}
\def\bbeta{\bar{\beta}}

\newtheorem{theorem}{Theorem}[section]

\newtheorem{prop}[theorem]{Proposition}

\begin{document}

\title{Loop Quantum Gravity Boundary Dynamics and $\SL(2,\C)$ Gauge Theory }

\author{{\bf Etera R. Livine}}\email{etera.livine@ens-lyon.fr}
\affiliation{Univ. Lyon, Ens de Lyon, Univ. Claude Bernard, CNRS, LPENSL, 69007 Lyon, France}

\date{\today}

\begin{abstract}


In the context of the quest for a holographic formulation of quantum gravity, we investigate the basic boundary theory structure  for loop quantum gravity. In 3+1 space-time dimensions, the boundary theory lives on the 2+1-dimensional time-like boundary and is supposed to describe the time evolution of the edge modes living on the 2-dimensional boundary of space, i.e. the space-time corner. Focusing on ``electric'' excitations -quanta of area- living on the corner, we formulate their dynamics in terms of classical spinor variables and we show that the coupling constants of a polynomial Hamiltonian can be understood as the components of a background boundary 2+1-metric. This leads to a deeper conjecture of a correspondence between boundary Hamiltonian and boundary metric states.
We further show that one can reformulate the quanta of area data in terms of a $\SL(2,\C)$ connection, transporting the spinors on the boundary surface and whose $\SU(2)$ component would define ``magnetic'' excitations (tangential Ashtekar-Barbero connection), thereby opening the door to writing the loop quantum gravity boundary dynamics as a 2+1-dimensional $\SL(2,\C)$ gauge theory.

\end{abstract}

\maketitle
\tableofcontents


%
%
%


\vspace*{10mm}
With the raise of the holographic principle as one -if not the- main guide towards a theory of quantum gravity, a very active line of research is presently the study of  boundary theories induced by general relativity, either asymptotically as in the AdS/CFT correspondence and the analysis of soft modes, or in a quasi-local fashion  through the investigation of edge modes living on space-time corners and boundary current algebras for finite regions of space-time. It seems that the bulk-boundary relation is an essential cornerstone of the theory; it is not the mere propagation from boundary conditions to bulk field configurations, but the realization that representing bulk observables as boundary charges encodes the symmetries and renormalization flow of quantum gravity.

The goal of this short paper is to tackle this issue in the loop quantum gravity framework. We will not study the corner structure or edge modes of classical general relativity and their possible quantization as e.g. in \cite{Freidel:2014qya,Freidel:2015gpa,Donnelly:2016auv,Freidel:2016bxd,DePaoli:2018erh,Freidel:2019ees,Freidel:2019ofr,Harlow:2019yfa,Takayanagi:2019tvn,Freidel:2020xyx,Freidel:2020svx,Freidel:2020ayo}, but  we focus instead on the question: what kind of boundary theory can the standard loop quantum gravity formalism support? This is a crucial point to address in order to understand what type of boundary conditions can one impose in loop quantum gravity, whether or not the standard framework of loop quantum gravity should be extended and quantum states of geometry enriched with more information, as proposed for instance in \cite{Bahr:2015bra,Charles:2016xwc,Delcamp:2018efi,Freidel:2018pvm,Livine:2019cvi,Freidel:2019ees}, and what type of boundary dynamics could one hope to translate from the classical setting to the quantum realm of loop gravity's spin network states.

In the canonical framework of loop quantum gravity based on a 3+1 splitting of space-time in terms of a 3d spatial hypersurface evolving in time, we focus on the space-time corner defined by the hypersurface's boundary. Loop quantum gravity defines quantum states of the bulk geometry as spin networks, which are polymer structures consisting of (embedded) graphs dressed with algebraic data from the representation theories of the Lie groups $\SU(2)$ (and possibly of the  $\SL(2,\C)$ Lorentz group and their  quantum group deformations). Assuming the spatial boundary to have the topology of a 2-sphere, we consider a spin network state puncturing the corner, as drawn on fig.\ref{fig:puncturedcorner}, similar to the lightning within a plasma globe, thus leading to the boundary data of the algebraic data carried by the links cut by the boundary. The corner thus carries a certain number of $\SU(2)$ spin states $|j_{i},m_{i}\ra$ at the quantum level, or of spinors $z_{i}\in\C^{2}$ at the classical level.
From this starting point, section \ref{sec:spinor} reviews the spinor phase space on space-time corners, according to the holomorphic reformulation of loop quantum gravity \cite{Freidel:2010aq,Freidel:2010bw,Freidel:2009nu,Borja:2010rc,Livine:2011gp,Speziale:2012nu,Livine:2013zha,Alesci:2016dqx,Calcinari:2020bft}, describes what type of spinor theory one could expect on the boundary for a given number of punctures and discusses the possible continuum limit into a spinor field theory\footnotemark{} living on the 2+1-d time-like boundary of space-time, which could be considered as a second quantization allowing to vary the number of punctures.
\footnotetext{
Such a description of the boundary data in terms of a continuous spinor field is of course reminiscent of the formulation of general relativity's boundary data on null surfaces in terms of spinors, as in e.g. \cite{Wieland:2016dbc,Wieland:2017zkf}. A possible relation between the two frameworks should clearly be investigated but is beyond the scope of the present work.
}

Section  \ref{sec:sl2c} is dedicated to re-writing the dynamics of boundary spinors as a $\SL(2,\C)$ gauge theory, having in mind the idea of a gravity-gauge duality linking the bulk and boundary dynamics for finite space-time regions. More precisely, we show how one can reformulate the kinematics and dynamics of a collection of spinors in terms of a discrete  $\SL(2,\C)$ connection. This leads us to the proposal of templates for loop quantum gravity boundary theories in the continuum  in terms of $\SL(2,\C)$ gauge connections living on the 2+1-d time-like boundary of space-time.
\begin{figure}[!htb]
%
%
%
%
%
%
%
\includegraphics[height=40mm]{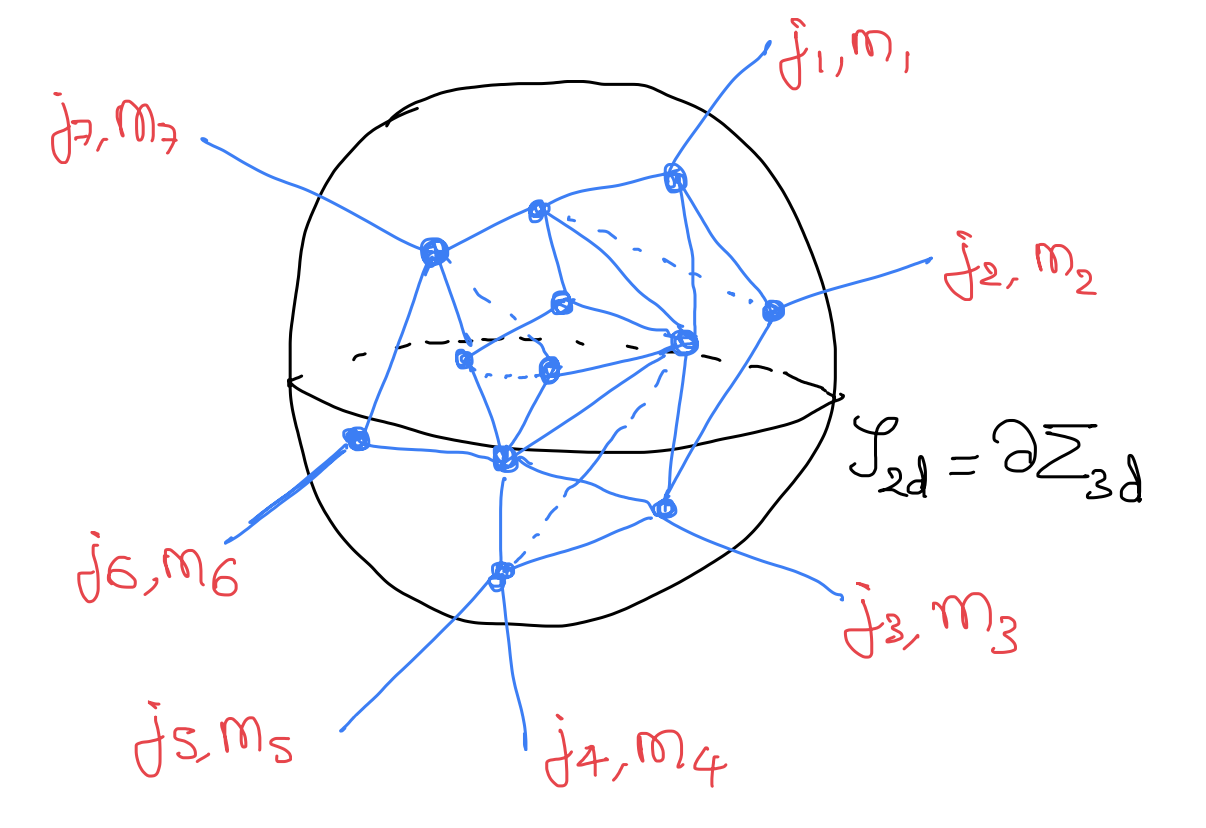}
	\caption{ Loop quantum gravity boundary data defined by the spin network puncturing the two-dimensional space-time corner $\cS=\pp \Sigma$: the spin network edges puncturing the boundary of the 3d space $\Sigma$ carry flux excitations, encoded as spin states $|j_{i},m_{i}\ra$ and defining quanta of area on the boundary surface.}
	\label{fig:puncturedcorner}
	\label{fig:boundarygraph}
\end{figure}
%

\section{Loop Quantum Gravity on Space-Time Corners}
\label{sec:spinor}


We consider here the simplest setting in loop quantum gravity: a region of the 3d space with the topology of a 3-ball bounded by a 2-sphere, whose quantum state is a spin network state whose links puncture the 2-sphere transversally (without  looking into the subtle possibility of links tangential to the boundary). The boundary data will thus consist in the fluxes carried by the open links going through the boundary. This means that we are focusing on flux excitations on the boundary, i.e. ``electric'' excitations, as illustrated in fig.\ref{fig:boundarygraph}, and discarding (as for now) the possibility of magnetic excitations and momentum excitations as presented in \cite{Freidel:2019ees,Freidel:2019ofr,Freidel:2020xyx,Freidel:2020svx}.

%
%
We do not assume any further a priori  boundary data, such as a graph-like  linking the punctures. For instance, such a boundary graph was introduced in \cite{Feller:2017ejs} as a notion of nearest neighbour (and thus of locality) on the boundary inherited from the bulk graph of the spin network (i.e. a notion of boundary locality as the projection of the near-boundary bulk locality). Taking into the magnetic excitations on the boundary also led to another notion of boundary graph, as a spin network living on the 2d boundary itself and carrying non-trivial holonomies of the connection along curves tangential to the boundary. This is a recurring idea in recent investigation of a necessary extension of the loop quantum gravity formalism,  as when representing twisted geometries as a moduli of $\SL(2,\C)$ flat connections \cite{Haggard:2015kew,Han:2016dnt}, double spin networks \cite{Charles:2016xzi}, Drinfeld tube networks \cite{Delcamp:2016yix,Delcamp:2018efi}, Poincar\'e charge networks \cite{Freidel:2019ees}. This should definitely be kept in mind in future investigation.
We will nevertheless see below in section \ref{subsec:dynamics} that a notion of boundary graph naturally arises when defining the dynamics for the boundary degrees of freedom.

Within those limitations of the present approach, the boundary data on a spatial boundary (or space-time corner) induced by the bulk geometry state in loop quantum gravity consists in the spin states carried by the punctures representing the flux excitations on the boundary. Each puncture state lives in the Hilbert space,
\be
\cH^{(1)}=\bigoplus_{j\in\f\N2}\cV_{j}\,,
\ee
where  $\cV_{j}$ is the $(2j+1)$-dimensional  Hilbert space carrying the spin-$j$ representation of the $\SU(2)$ Lie group. The Hilbert space for boundary states with $N$ punctures is simply the tensor product of $N$ copies of the one-puncture Hilbert space:
\be
\cH^{(N)}=\bigotimes_{i=1}^{N}\cH^{(1)}_{i}
\,.
\ee
And our goal is to describe the dynamics of the boundary flux excitations represented as the spin states carried by the punctures.
This is actually the same setting as if one would like endow an isolated horizon with microscopic dynamics in loop quantum gravity in the standard description of quantum horizons \cite{Ashtekar:1997yu,Ashtekar:2000eq,Domagala:2004jt,Agullo:2009zt,Asin:2014gta}. Indeed, a horizon state is entirely described as a tensor product of spin states . The boundary theory is usually to be trivial, i.e. the spin states are usually assumed not to interact, although the interested reader can see \cite{Feller:2017ejs} for a proposal of bulk-induced interaction on the horizon through a Bose-Hubbard exchange hamiltonian. Here we do not place ourselves in the restrictive setting of an (isolated) horizon but consider a general boundary surface without assuming any specific boundary condition. The present work would nevertheless be relevant when investigating the quantum gravity dynamics around black hole horizons.

Interestingly, the spin states can be seen as wave-functions over a complex 2-vector, or spinor $z=(z^{0},z^{1})\in\C^{2}$. This is a standard construction in mathematics for semi-simple Lie groups. 
The vector space $\cV_{j}$ consists in homogeneous polynomials in $z^{0}$ and $z^{1}$ of degree $2j$, with the standard basis state $|j,m\ra$ corresponding to the monomial $(z^{0})^{j+m}(z^{1})^{j-m}$. The $\SU(2)$ action on those polynomials results from the natural action of $\SU(2)$ group elements as 2$\times$2 matrices on the spinor $z$:
\be
\forall g\in\SU(2)\,, z\in\C^{2}\,,\qquad
g\triangleright z = \mat{cc}{a & b \\ c & d}\mat{c}{z^{0}\\z^{1}}\,\in\C^{2}
\,.
\ee
The scalar product is then given by the integration of the polynomials with respect to the Gaussian measure on $\C^{2}$ (up to a $j$-dependent factor)
or equivalently to the integration over the $3$-sphere of unit spinors satisfying $\la z|z\ra=|z^{0}|^{2}+|z^{1}|^{2}=1$. This simple correspondence is the foundation for the holomorphic reformulation of loop quantum gravity in terms of spinorial coherent states as proposed in \cite{Borja:2010rc,Livine:2011gp,Speziale:2012nu,Livine:2013zha,Alesci:2016dqx}.
As a consequence, we work with  boundary data on the space-time corner given as a collection of spinors $z_{i}\in\C^{2}$, one spinor variable per puncture. In this section, we explain how to endow those spinors with a dynamics. In the next section, we will show how to reformulate collection of spinors in terms of flat $\SL(2,\C)$ connections on the boundary.


\subsection{Punctures and Spinors on the Boundary}

%
%
%
%
%
%

The phase space for $N$ spin network punctures on the boundary consists in $N$ spinors $z_{i}\in\C^{2}$ with $i=1..N$ provided with a canonical Poisson bracket:
\be
\label{zzbracket}
\{z^{A}_{i},\bz^{B}_{j}\}=-i\delta^{AB}\delta_{ij}
\,,\qquad
\{z^{A}_{i},z^{B}_{j}\}=\{\bz^{A}_{i},\bz^{B}_{j}\}=0
\,.
\ee
Each spinor $z_{i}\in\C^{2}$ defines a 3-vector $\vX_{i}\in\R^{3}$, identified as the flux associated to the puncture. Using bra-ket notation, the flux vectors are defined as:
\be
X_{i}^{a}=\f12\la z_{i}|\sigma^{a} |z_{i} \ra\,,
\ee
where the $\sigma^{a}$, with $a=1..3$, are the three Pauli matrices, normalized such that $(\sigma^{a})^{2}=\id$. Each flux vector forms a $\su(2)$ Lie algebra:
\be
\{X_{i}^{a},X_{j}^{b}\}=\delta^{ab}\delta_{ij}\,.
\ee
This is simply the Schwinger presentation of the $\su(2)$ algebra.
Indeed, upon quantization, we promote each spinor component to a quantum harmonic oscillator:
\be
\left|
\begin{array}{lcl}
z^{A}_{i}&\rightarrow & a^{A}_{i}\\
\bz^{A}_{i}&\rightarrow & a^{\dagger A}_{i}
\end{array}
\right.
\qquad\textrm{with}\quad
[a^{A}_{i},a^{\dagger B}_{j}]=1\,.
\ee
Then the flux vectors become the $\su(2)$ Lie algebra generators for the $\SU(2)$ representation attached to each puncture considered as an open end of the spin network links:
\be
X^{a}_{i}\rightarrow  J^{a}_{i}\,,\quad
J^{a}_{i}=J^{\dagger a}_{i}
\qquad\textrm{with}\quad
[J^{a}_{i},J^{b}_{j}]=i\delta_{ij}\eps^{abc}J^{c}_{i}\,.
\ee
These $\SU(2)$ representations are all independent and we recover the Hilbert space $\cH^{(N)}$ for $N$ spin network punctures on the boundary.
The $\SU(2)$ Casimir of each puncture, $\cC_{i}=J^{a}_{i}J^{a}_{i}$, corresponds to the quantization of the squared flux vector norm $|\vX_{i}|^{2}$ and gives the spin $j_{i}$ carried by the puncture by the usual Casimir expression, $\cC_{i}=j_{i}(j_{i}+1)$.
Actually, in the Schwinger presentation, there is an operator corresponding to the quantization of the flux vector norm $|\vX_{i}|$ and giving directly the spin $j_{i}$:
\be
|\vX_{i}|=\f12\la z_{i}|z_{i} \ra
\rightarrow
j_{i}=\f12\sum_{A=0,1}a^{\dagger A}_{i}a^{ A}_{i}
\,.
\ee
This provides the spin $j_{i}$ with the interpretation as a quantum of area $\cA_{i}=j_{i}\ell_{Planck}^{2}$ on the boundary surface in Planck units (suitably renormalized by the Immirzi parameter).
For details on the Schwinger presentation and the derivation of the standard spin basis $|j,m\ra$, the interested reader can refer to \cite{Freidel:2009ck}. Furthermore, for details on how the classical spinor variables naturally label a system of coherent states minimizing the uncertainty relations associated to the $\su(2)$ commutators and transforming consistently under the $\SU(2)$ action, the interested reader can refer to \cite{Freidel:2010tt,Livine:2013tsa}. These $\SU(2)$ coherent states are a basic tool in the definition and construction of EPRL-like spin foam amplitudes for the dynamics of spin networks \cite{Livine:2007vk,Livine:2007ya,Freidel:2007py,Dupuis:2011fz,Dupuis:2011wy,Freidel:2012ji,Speziale:2012nu,Banburski:2014cwa}.

A final subtlety is that a spinor $z\in\C^{2}$, with four real components, contains more information than the flux vector $\vX\in\R^{3}$, with three real components. The extra data is the phase of the spinor. Indeed the flux vector is invariant under phase shift of the spinor:
\be
\forall \phi\in\R\,,\qquad \vX(e^{i\phi}z)=\vX(z)
\,.
\ee
This phase degree of freedom is called the {\it twist angle} and identified as a measure of the extrinsic curvature (of the canonical hypersurface) at the puncture  \cite{Freidel:2010aq,Freidel:2010bw}. It is the variable canonically conjugate to the flux vector norm $\vX$, i.e. to the puncture area, and thus to the spin at the quantum level. Without this phase freedom, one can not define a phase space with varying puncture area, i.e. quantized into a Hilbert space allowing arbitrary spin superpositions (see e.g. \cite{Livine:2013tsa} for more mathematical details).

\medskip

In the rest of the paper, we will not work at the quantum level with spin states but will work in the simpler classical setting of the boundary spinorial phase space. This allows to describe the boundary theory in the space-time corner through action principles defined in terms of the spinor variables. These action principles can then be quantized at finite number of punctures, for  fixed $N$, or quantized and renormalized as quantum field  theories in the continuum limit of infinite refinement  $N\rightarrow\infty$.

\subsection{Spinor dynamics on the 2+1-d time-like boundary}
\label{subsec:dynamics}

The present goal is to define dynamics for the  degrees of freedom on the spatial boundary of spin network states, that is for the spinor variables attached to the punctures on the boundary.
These boundary degrees of freedom are initially defined on the 2d boundary of the canonical 3d hypersurface and then evolve along the 2+1-d time-like boundary of the 3+1-d space-time.

Let us start with the case of a fixed number of punctures $N$. A natural action principle for the boundary dynamics is, first to take into account the canonical Poisson bracket \eqref{zzbracket} between spinor components, second define a relativistic dynamics through a Hamiltonian constraint accounting for the invariance of the theory under time reparametrizations. This yields the general ansatz for the boundary action:
\be
\label{actionN}
S_{\pp}[\{z_{k}(t)\}_{k}]
\,=\,
\int \dd t\,\left[
-i\sum_{k}\la z_{k}|\dd_{t} z_{k}\ra
-\cN\cH[\{z_{k}\}]
\right]\,,
\ee
where $\cN$ is a lapse variable and $\cH[\{z_{k}\}]$ is the to-be-specified Hamiltonian constraint.
At that point, there are two standard strategies:
\begin{itemize}
\item {\bf A.} one can analyze the Hamiltonian framework of general relativity with boundary terms, understand the boundary symmetry algebra and derive the pool of possible boundary theories depending on the chosen boundary conditions, discretize and/or quantize the resulting boundary dynamics in order to translate it in terms of discrete geometry, spin network states and boundary surface excitations;

\item {\bf B.} one can alternatively work within the already defined mathematical framework of loop quantum gravity, derive the natural boundary variables, understand the pool of possible dynamics that can be defined for this boundary data and investigate the resulting physics.

\end{itemize}

Ultimately, one of course wish for a convergence of those approaches. Here we will focus on the latter strategy (B) and postpone the comparison with the first strategy (A) to future (but necessary) investigation.
From this viewpoint, the natural path to follow is to investigate what kind of Hamiltonian $\cH[\{z_{k}\}]$ we should, or could, define for the boundary dynamics, and then, later, analyze the coarse-graining, renormalization flow, possible fixed points and universality classes dynamics. Going step by step, let us understand what makes a good ansatz for the Hamiltonian constraint.
Once we have chosen basic variables, here the spinors, it is natural to proceed to a Taylor expansion of the functionals and thus consider polynomial ansatz (of increasing power) for the Hamiltonian$\cH[\{z_{k}\}]$.
If the quadratic ansatz, corrected by potential higher order terms, does not yield expected or realistic physics, then this usually indicates that we have made the wrong choice of basic variables and that we should very likely consider a different phase for spin networks (for example, consider a type of condensate of spin networks, as proposed in the group field theory approach \cite{Oriti:2016acw,Oriti:2018qty,Carrozza:2020akv}) and the corresponding different boundary data for loop quantum gravity on space-time corners.

We naturally require the action to be invariant to be invariant under global $\SU(2)$ transformations acting simultaneously on all the spinor:
\be
|z_{k}\ra\,\mapsto g\,|z_{k}\ra\,,\qquad g\in\SU(2)\,.
\ee
This descends from the $\SU(2)$ local gauge invariance of canonical general relativity (formulated  in terms of vierbein/connection variables) translated to a single overall $\SU(2)$ reference frame for the whole spatial boundary. This can be seen as a gauge-fixing of the $\SU(2)$ local gauge invariance of the bulk spin network state down to a $\SU(2)$ gauge invariance for the boundary data rooted at a chosen boundary vertex (see \cite{Livine:2013gna} for more details on the bulk-to-boundary projection of spin network states from a coarse-graining perspective).
We will discuss later in the section \ref{sec:gauge} the possibility of imposing a local $\SU(2)$ gauge invariance on the boundary and its relation to magnetic boundary degrees of freedom.
As for now, imposing this global $\SU(2)$ gauge  invariance on the boundary  automatically removes the possibility of linear terms in the boundary Hamiltonian.

\medskip

At the quadratic level, we can have local potential terms in $\la z_{k}| z_{k}\ra$ for each puncture. As there is not an obvious reason to prefer one puncture over another, this lead to a single term for the Hamiltonian constraint:
\be
\cH=\beta_{2}\sum_{k}\la z_{k}| z_{k}\ra +\dots\,,
\ee
where $\beta_{2}$ is a to-be-specified coupling constant. This first term of the Taylor expansion is simply the total boundary area $\cA[\{z_{k}\}]=\sum_{k}\la z_{k}| z_{k}\ra$. Then we can introduce diffusion/propagation terms as non-local terms coupling punctures together.
There are two types of such terms, scalar products $\la z_{k}| z_{l}\ra$ and scalar products $[ z_{k}| z_{l}\ra$, where we have used the notation introduced in \cite{Freidel:2010bw,Freidel:2010tt} to denote the dual spinor:
\be
|z\ra =\mat{c}{z^{0}\\z^{1}}
\,,\qquad
|z]=\mat{c}{-\bz^{1}\\ \bz^{0}}
=
\mat{cc}{0 & -1\\+1 & 0}\mat{c}{\bz^{0}\\ \bz^{1}}
\,,
\ee
\be
\la z| w\ra=\bz^{0} w^{0}+\bz^{1} w^{1}
\,,\qquad
[z|w\ra=z^{0}w^{1}-z^{1}w^{0}\,.
\ee
The dual spinor transforms under the same $\SU(2)$ group action as the original spinor:
\be
\forall g\in\SU(2)\,,\quad
g|z\ra
=\mat{cc}{\alpha &-\bbeta \\ \beta & \balpha}\mat{c}{z^{0}\\z^{1}}
=\mat{c}{  \alpha z^{0}-\bbeta z^{1}\\ \balpha z^{1}+\beta z^{0}}
=|gz\ra
\ee
\be
\Rightarrow\quad
g|z]=\mat{cc}{\alpha &-\bbeta \\ \beta & \balpha}\mat{c}{-\bz^{1}\\ \bz^{0}}
=\mat{c}{-\alpha\bz^{1}-\bbeta\bz^{0}\\  \balpha\bz^{0}-\beta\bz^{1}}
=|gz]
\,.
\ee
A direct consequence is that the scalar products  $\la z_{k}| z_{l}\ra$ and $[ z_{k}| z_{l}\ra$ are the only quadratic polynomials which are invariant under the global $\SU(2)$ action on the spinors. The combinations  $\la z_{k}| z_{l}\ra$ commute with the total area $\cA[\{z_{i}\}]=\sum_{i}\la z_{i}| z_{i}\ra$ and form a $\u(N)$ Lie algebra as shown in \cite{Girelli:2005ii,Freidel:2009ck}:
\be
\{\la z_{k}| z_{l}\ra,\cA\}=0
\,,\qquad
\{\la z_{k}| z_{l}\ra,\la z_{m}| z_{n}\ra\}=
i\big{(}\delta_{kn}\la z_{m}| z_{l}\ra-\delta_{lm}\la z_{k}| z_{n}\ra\big{)}
\,.
\ee
Therefore the Hamiltonian flow that they generate can be integrated as $\U(N)$ transformations that  preserves the total surface area of the boundary \cite{Girelli:2005ii,Freidel:2009ck,Livine:2013tsa}.

On the other hand, the combinations $[ z_{k}| z_{l}\ra$ do not commute with the total area. There Hamiltonian flow actually decreases the  boundary area, while their complex conjugates $\la z_{k}| z_{l}]$ generate a Hamiltonian flow increasing the  boundary area:
\be
\{[ z_{k}| z_{l}\ra,\cA\}=+i \la z_{k}| z_{k}\ra+i \la z_{l}| z_{l}\ra
\,,\qquad
\{\la z_{k}| z_{l}],\cA\}=-i \la z_{k}| z_{k}\ra-i \la z_{l}| z_{l}\ra
\,.
\ee
These are thus creation and annihilation of entangled area quanta between the two punctures $k$ and $l$. The scalar product $[ z_{k}| z_{l}\ra$ is holomorphic in the spinors, while  its complex conjugate $\la z_{k}| z_{l}]$ is anti-holomorphic. Their Poisson brackets do not close on their own. Nevertheless, together with the area-preserving scalar products, they form a closed $\so^{*}(2N)$ Lie algebra \cite{Freidel:2010tt,Girelli:2017dbk}:
\be
\begin{array}{l}
\{[ z_{k}| z_{l}\ra,[ z_{m}| z_{n}\ra\}=\{\la z_{k}| z_{l}],\la  z_{m}| z_{n}]\}=0
\,,\vspace*{1mm}\\
\{[ z_{k}| z_{l}\ra,\la z_{m}| z_{n}]\}=-i\big{(}
\delta_{lm}\la z_{n}|z_{k}\ra-\delta_{ln}\la z_{m}|z_{k}\ra-\delta_{km}\la z_{n}|z_{l}\ra+\delta_{kn}\la z_{m}|z_{l}\ra
\big{)}
\,,\vspace*{1mm}\\
\{\la z_{k}| z_{l}\ra,[ z_{m}| z_{n}\ra\}=
i\big{(}
\delta_{kn}[ z_{m}|z_{l}\ra-\delta_{km}[ z_{n}|z_{l}\ra
\big{)}
\,.
\end{array}
\ee
The Hamiltonian flow of all those quadratic couplings between punctures can thus be integrated as a $\SO^{*}(2N)$ Lie group flow.

The resulting quadratic ansatz for the boundary Hamiltonian constraint, including the homogeneous local potential term and the puncture interaction terms, reads:
\be
\cH=
\beta_{2}\sum_{k}\la z_{k}| z_{k}\ra
+
\gamma_{2}\sum_{k,l} C_{kl}\la z_{k}| z_{l}\ra
+
\tilde{\gamma}_{2}\sum_{k,l} \Big{(}D_{kl}[ z_{k}| z_{l}\ra +\overline{D}_{lk}\la z_{k}| z_{l}]\Big{)}
+\dots\,,
\ee
where we have added the coupling constants $\gamma_{2}$ and $\tilde{\gamma}_{2}$ for respectively the area-preserving interactions and the area-changing interactions. The matrix $C$ can be assumed symmetric in order to ensure that the Hamiltonian be real. It generates an exchange of boundary area quanta between punctures. We refer to this term as the ``exchange Hamiltonian''. The coefficients $C_{kl}$ give the strength of the coupling between the two punctures $k$ and $l$. As suggested in \cite{Feller:2017ejs}, one could restrict this exchange Hamiltonian to nearest neighbour interactions. This is done by drawing a boundary graph between the punctures and taking the matrix $C$ to be the adjacency matrix of this graph. The boundary graph does not necessarily need to be planar, although this is a natural choice if the boundary topology is that of a 2-sphere. This boundary graph is interpreted as a notion of locality on the boundary, which can be thought as given a priori as induced by the bulk spin network structure (as in \cite{Feller:2017ejs}), or as a choice of boundary background structure, or defined a posteriori from the choice of Hamiltonian.
The local potential term can be considered simply as the diagonal terms of exchange matrix $C$.

The matrix $D$ encodes the coupled creation and annihilation of quanta of area on the boundary. We refer to this term as the ``expansion Hamiltonian''. Together the matrices $C$ and $D$ define the quadratic truncation of the Hamiltonian constraint. The exponentiated flow of this quadratic Hamiltonian can be integrated in terms of $\SO^{*}(2N)$ group elements. If we remove the expansion term and solely focus on the exchange Hamiltonian, the flow can be then more simply integrated in terms of $\U(N)$ group elements, which describe the deformation of the boundary surface at constant area \cite{Freidel:2009ck,Livine:2013tsa}.

This quadratic ansatz already seems rich enough, especially for the purpose of comparing the importance of boundary wave propagation (generated by the exchange Hamiltonian) versus  the expansion/shrinking of the boundary surface.
However, since the local potential term can be entirely re-absorbed in the quadratic exchange term, it seems natural to push its Taylor expansion to the next order and introduce quartic local terms $\la z_{k}| z_{k}\ra^{2}$ for each puncture, thus yielding a more complete ansatz for the boundary Hamiltonian constraint:
\be
\label{boundaryH}
\cH=
\beta_{2}\sum_{k}\la z_{k}| z_{k}\ra
+
\gamma_{2}\sum_{k,l} C_{kl}\la z_{k}| z_{l}\ra
+
\tilde{\gamma}_{2}\sum_{k,l} \Big{(}D_{kl}[ z_{k}| z_{l}\ra +\overline{D}_{lk}\la z_{k}| z_{l}]\Big{)}
+
\beta_{4}\sum_{k}\la z_{k}| z_{k}\ra^{2}
+\dots\,,
\ee
This new term moves the Hamiltonian outside of the $\so^{*}(2N)$ Lie algebra and the exponentiated flow can not be simply expressed in terms of $\SO^{*}(N)$ group elements. Putting the expansion term aside by setting $\tilde{\gamma}_{2}$ to zero, and assuming the matrix $C$ to be the adjacency matrix of a boundary graph, the truncation to area-preserving terms defines a {\it Bose-Hubbard} Hamiltonian, with a local potential $\sum_{k}\la z_{k}| z_{k}\ra^{2}$ balanced against an exchange term of quanta between nearest neighbours $\sum_{(k,l)}\la z_{k}| z_{l}\ra$. We call this the Bose-Hubbard truncation for the boundary dynamics. The quartic potential term will clearly affect the propagation of waves on the boundary. Beside possible Anderson localization phenomena, this will allow to study the propagation of perturbations on the boundary surface -ballistic versus diffusive- and the resulting relaxation towards homogeneous equilibrium (or, even, the possible instability of the homogeneous configuration).
It was even speculated in \cite{Feller:2017ejs} that such a Bose-Hubbard model applied to black hole horizons should lead to a phase transition between classical black hole with fast relaxation of the event horizon under local perturbations and a quantum regime for microscopic black holes where horizon perturbations will diffuse in a slower fashion making  the in-falling information more visible to the exterior observer. More generally, it would be interesting to identify which phase of the generic ansatz given above could correspond to  an isolated horizon and to a black hole or to other type of boundary conditions studied in general relativity.

We believe that this ansatz \eqref{boundaryH} for a boundary Hamiltonian constraint at fixed number of punctures $N$ is a good starting point for a systematic study of the boundary dynamics in loop quantum gravity. It could already lead to exciting new phenomena and predictions for quantum gravity.

\subsection{Boundary  Field Dynamics}
\label{subsec:field}

We have provided a natural ansatz for the loop quantum gravity boundary dynamics of flux excitations for a finite fixed number $N$ of punctures on the space-time corner. Geometrically this corresponds to a deep quantum regime where the boundary surface is made of $N$ elementary surface patches carrying quanta of area. In order to compare with the classical analysis of boundaries in general relativity (see the recent work in e.g. \cite{Harlow:2019yfa,Takayanagi:2019tvn,Freidel:2020xyx,Freidel:2020svx,Freidel:2020ayo}), it is appropriate to look at the classical regime for these boundary degrees of freedom. To this purpose, we consider the na\"ive continuum limit defined by the infinite refinement of the boundary surface, sending the number of punctures to infinity $N\rightarrow+\infty$. We call this ``na\"ive'' because we will not look at the dynamics of coherent states and seek to define a semi-classical regime at high energy and high action, but we will instead consider the boundary action \eqref{actionN}-\eqref{boundaryH} at fixed $N$ described in the previous section as the discretization of a continuous field theory. From a 2nd quantization viewpoint, the punctures can be thought of as ``boundary particles'' and the field theory considered as describing the regime where the number of punctures can (arbitrarily) fluctuate.

Introducing a 2d coordinate system $x^{k=1,2}$ on the boundary sphere, the sum over punctures becomes an integral over the 2-sphere. Assuming that the exchange and expansion terms in the finite $N$ Hamiltonian \eqref{boundaryH} are defined on terms of nearest neighbours on a (planar) lattice living on the 2-sphere, and taking the infinite refinement of this boundary lattice, we get an action for a continuous spinor field $z(x^{k})\in\C^{2}$ living on the 2+1 time-like boundary of space-time parametrized by the time coordinate $t$ and the 2d coordinates $x^{k}$:
\be
\label{actionfieldH}
S_{\pp}[z,\bz,\cN]
=
\int \rd t\,\int \rd^{2} x\, \bigg{[}
i\la z|\pp_{t} z\ra
-\cN\Big{[}
i \gamma^{k} \la z |\pp_{k} z \ra
+ \f i 2\left(\tilde{\gamma}^{k}[ z |\pp_{k} z \ra
+\overline{\tilde{\gamma}^{k}}\la z |\pp_{k} z ]\right)
+
\sum_{n}\beta_{n}\la z | z\ra^{n}
\Big{]}
\bigg{]}\,.
\ee
This is a first order action principle with kinetic terms and local potential. It is real (up to boundary terms\footnotemark) and has coupling constants $\beta_{n}$ for the Taylor expansion of the potential and $\gamma^{k},\tilde{\gamma}^{k}$ defining the geometric background.
\footnotetext{
Boundary terms of this boundary action live on the corner of the 2+1-dimensional time-like boundary of space-time, i.e. on 1d boundaries on the 3d spatial slices, i.e. on the contour around the punctures. Such terms will ultimately be relevant and should be studied in more details. From a topological field theory point of view, cells of every dimension carry algebraic data whose type depend on the dimensionality and represent the boundary charges generated by the corresponding boundary terms in the action.
}
Having usual derivative terms such as $\la z |\pp z\ra=\bz^{0}\pp z^{0}+\bz^{1}\pp z^{1}$ and holomorphic derivative terms such as $[ z |\pp z \ra=z^{0}\pp z^{1}-z^{1}\pp z^{0}$ can feel a little awkward. It could nevertheless become more natural if writing it as a four-dimensional spinor $\Psi=(|z\ra,|z])$ encompassing both the 2-spinor and its dual spinor. This 4-spinor does not however have any special meaning or role, so we do not pursue this possibility.


In order to better understand the meaning of this field theory, an interesting step is go back to a Lagrangian formulation. Since the field here is complex, we first separate its real and imaginary parts and then compute the inverse Legendre transform. This will yield a second order Lagrangian for the boundary theory. To illustrate the procedure, we start by presenting the example of a complex field instead of the spinor field.
Let us consider the following 1+1-d action in its Hamiltonian form, truncated to quadratic terms:
\be
S[z(x)]=\int \rd t\,\rd x\,
\Big{[}
i\bz \pp_{t} z -\cN  \f\beta2 z\bz- \cN i\gamma \bz \pp_{x} z
\Big{]}
\,.
\ee
Writing $z=(q+ip)/\sqrt{2}$, this action reads (up to total derivative terms, which we set aside):
\be
S[z(x)]=\int \rd t\,\rd x\,
\Big{[}
p \pp_{t}q -\cN \f\beta2 (q^{2}+p^{2})- \cN \gamma p \pp_{x} q
\Big{]}
\,.
\ee
We can consider $q$ as the configuration field, while $p$ is its conjugate momentum. Solving for the field $p$ yields:
\be
\cN \beta p=\pp_{t}q-\cN\gamma \pp_{x}q
\,.
\ee
Plugging this expression into the action gives its Lagrangian formulation:
\be
S[q(x),\pp_{t}q(x)]
=
\f1\beta\int \rd t\,\rd x\,
\Big{[}
\f1{2 \cN}(\pp_{t}q)^{2}-\gamma \pp_{t}q\pp_{x}q+\cN\f{\gamma^{2}}{2}(\pp_{x}q)^{2}-\cN\f{\beta^{2}}2 q^{2}
\Big{]}
\,.
\ee
Here we immediately recognise the action for a (real) massive scalar field living on a curved 1+1-d metric written in its ADM form with the lapse factor in the time direction.
Indeed, if we write a d+1-dimensional metric $h$ for a space-time foliation in terms of space-like slices according its ADM form,
\be
\rd s^{2}
=
h_{\mu\nu}\rd x ^{\mu}\rd x ^{\nu}
=
(\cN^{2}-\cN^{k}\cN_{k})\rd t^{2}-2\cN_{k}\rd t\rd x^{k}-c_{kl}\rd x^{k}\rd x^{l}
\,,
\ee
in terms of the lapse $\cN$, shift vector $\cN^{k}$ and the d-dimensional space metric $c$,
the action for a (real) massive scalar field $\phi$ reads:
\be
S_{h}[\phi]=
\int \rd t\rd^{d}x\,
\sqrt{c}\,
\left[
\f1{\cN} (\pp_{t}\phi)^{2}
-2\f{\cN^{k}}{\cN^{2}}\pp_{t}\phi\pp_{k}\phi
-\left(c^{kl}-\f{\cN^{k}\cN^{l}}{\cN^{2}}\right)\pp_{k}\phi\pp_{l}\phi
-m^{2}\cN\phi^{2}
\right]
\,.
\ee
For a 1+1-dimensional field, this leads to the identification of the coupling constants of the spinor field hamiltonian in terms of the space-time metric components and the field properties: the mass $m$ is identified to the quadratic potential coupling $\beta$ and the exchange coupling $\gamma$ is identified (up to a numerical factor) to the normalized shift vector  $\cN^{x}/\cN$.

Following the same procedure from the spinor field action \eqref{actionfieldH} in its Hamiltonian form truncated to quadratic terms (i.e. discarding the quartic potential and higher order terms), we decompose the two spinor components in their real and imaginary parts, $z^{A}=(q^{A}+ip^{A})/\sqrt{2}$, leading to a second order action for two coupled scalar fields, $\phi^{A}=q^{A}$ with $A=0,1$. Assuming for the sake of simplicity for the couplings $\tilde{\gamma}^{k}$ are real, the inverse Legendre transform of the action \eqref{actionfieldH} gives the following lagrangian:
\beq
S_{\pp}[\phi^{0},\phi^{1}]&=
\beta^{-1} \displaystyle\int \rd t\rd^{2}x\,
\Bigg{[}&
\f1{2\cN}(\pp_{t}\phi^{0})^{2}+\f1{2\cN}(\pp_{t}\phi^{1})^{2}
-\gamma^{k}\pp_{t}\phi^{0}\pp_{k}\phi^{0}-\gamma^{k}\pp_{t}\phi^{1}\pp_{k}\phi^{1}
-\tilde{\gamma}^{k}\pp_{t}\phi^{0}\pp_{k}\phi^{1}+\tilde{\gamma}^{k}\pp_{t}\phi^{1}\pp_{k}\phi^{0}\nn\\
&&+\f{\cN}2(\gamma^{k}\tilde{\gamma}^{k'}+\gamma^{k'}\tilde{\gamma}^{k})(\pp_{k}\phi^{0}\pp_{k'}\phi^{0}+\pp_{k}\phi^{1}\pp_{k'}\phi^{1})
+2\cN(\gamma^{k}\tilde{\gamma}^{k'}-\gamma^{k'}\tilde{\gamma}^{k})\pp_{k}\phi^{0}\pp_{k'}\phi^{1}
\nn\\
&&-\f{\cN}2\beta^{2}(\phi^{0})^{2}-\f{\cN}2\beta^{2}(\phi^{1})^{2}
\Bigg{]}\,.
\eeq
We recognise the action\footnotemark{} for a pair of coupled massive scalar field -or equivalently a 2-component scalar field $\phi^{A}$- on the 2+1-dimensional time-like boundary written in its ADM form, with an identification of the shift $\cN^{k}$ and corner metric $c$ with the spinor hamiltonian coupling constants $\gamma^{k},\tilde{\gamma}^{k}$.
\footnotetext{
One might wonder about the corner metric factor $\sqrt{c}$. It does not appear in the continuum limit of the presently studied discrete spinor action. Nevertheless, it can easily be made to appear by adding a puncture-dependent weight factor to the sum over punctures. This might indicate that this is a necessary ingredient of the discrete formulation in order to obtain a properly covariant continuum limit. Another possibility, which we won't study here, is to extract the 2d metric density from the norm of the spinors themselves, as hinted by the balance equation between 2d metric and flux norm derived in \cite{Freidel:2015gpa,Freidel:2018pvm}.
}
This little exercise of going back to the lagrangian from the postulated hamiltonian boundary dynamics shows that the spinor variables associated to the boundary punctures do not simply lead to an actual spinor field but can be written as scalar fields in the continuum limit. Here we do not refer to the statistics of the field, but to the type of covariant derivative to which it couples. Indeed, due to the intricate coupling between its two components, the field $\phi$ could still  acquire non-trivial statistics (at the quantum level), but this would require a detailed analysis of the physics of this boundary action.

Although the derivation of the continuum field theory in its hamiltonian form from the spinor dynamics at finite number of punctures, and then the derivation of the lagrangian field theory by an inverse Legendre transform, are straightforward analytical steps, this procedure faces a few hurdles:
\begin{itemize}
\item the role of higher order potential terms: 

In the spinor dynamics for a fixed number of boundary punctures $N$, the higher order potential terms, such the quartic Bose-Hubbard coupling, typically balance the exchange of quanta between punctures and modulate the propagation of waves on the boundary surface. Such terms, of the type $\la z|z\ra^{n}$ for $n\ge2$, involve higher powers of both the scalar field components $q^{A}$ and their conjugate momenta $p^{A}$. These terms not only lead to a more complicated relation between momenta and derivatives of the field, making much harder to perform explicitly the inverse Legendre transform, but they also involve high powers of the field derivative, leading most likely to a higher order lagrangian. One should probably seek inspiration from the analysis of the semi-classical regime and continuum limit of the Bose-Hubbard model in order to better understand the physics created by those higher order potential terms. 

\item the $\SU(2)$ invariance and the coupling between the two field components $\phi^{0}$ and $\phi^{1}$:

The spinor hamiltonian \eqref{boundaryH} for fixed number of puncture $N$, as well as the continuous spinor  action \eqref{actionfieldH}, are both obviously invariant under global $\SU(2)$ transformations acting on the whole spatial slice, $|z\ra\mapsto g|z\ra$ for $g\in\SU(2)$. On the other hand, once splitting the spinor variables into real and imaginary part, respectively defining the scalar field $\phi^{A}$ and its conjugate momentum, this obvious character of the $\SU(2)$ action is lost. Indeed, $\SU(2)$ transformations do not act simply on the scalar field components, but become Bogoliubov canonical transformations\footnotemark{} mixing the scalar field and its momentum field. Nonetheless, the action is still invariant under those transformations and this gauge invariance somehow reflects into the intricate structure of the coupling between the two field components $\phi^{0}$ and $\phi^{1}$. It would definitely be interesting to investigate further this symmetry, for instance understand the expression of the induced Noether charges in terms of the boundary scalar field.
%
\footnotetext{
$\SU(2)$ group elements naturally act as 2$\times$2 unitary matrices on complex 2-vector in its fundamental representation.  However, to our knowledge, the resulting $\SU(2)$ action on the real and imaginary parts of the 2-vector does not any straightforward geometrical or physical interpretation:
\be
\mat{cc}{\alpha &-\bar{\beta}\\ \beta & \bar{\alpha}}\mat{c}{z^{0}\\z^{1}}
=
\mat{cc}{e^{i\vphi}\cos\theta &-e^{-i\psi}\sin\theta\\ e^{i\psi}\sin\theta & e^{-i\vphi}\cos\theta}\mat{c}{q^{0}+ip^{0}\\q^{1}+ip^{1}}
=
\mat{c}{\tilde{q}^{0}+i\tilde{p}^{0}\\\tilde{q}^{1}+i\tilde{p}^{1}}
\,,\nn
\ee
\be
\textrm{with}\quad\left|
\begin{array}{lcl}
\tilde{q}^{0}&=&\cos\theta\cos\psi q^{0}-\cos\theta\sin\psi p^{0}-\sin\theta\cos\psi q^{1}-\sin\theta\sin\psi p^{1} \,,\\
\tilde{p}^{0}&=&\cos\theta\sin\psi q^{0}+\cos\theta\cos\psi p^{0}+\sin\theta\sin\psi q^{1}-\sin\theta\cos\psi p^{1} \,,\\
\tilde{q}^{1}&=&\sin\theta\cos\psi q^{0}-\sin\theta\sin\psi p^{0}+\cos\theta\cos\psi q^{1}+\cos\theta\sin\psi p^{1} \,,\\
\tilde{p}^{1}&=&\sin\theta\sin\psi q^{0}+\sin\theta\cos\psi p^{0}-\cos\theta\sin\psi q^{1}+\cos\theta\cos\psi p^{1} \,.
\end{array}
\right.
\nn
\ee
This is a canonical transformation, leaving the canonical Poisson brackets invariant $\{\tilde{q}^{0},\tilde{p}^{0}\}=\{\tilde{q}^{0},\tilde{p}^{0}\}=1$ and $\{\tilde{q}^{0},\tilde{p}^{0}\}=0$, but mixing linearly the configuration variables and their momenta. It is thus a Bogoliubov transformation.
}

\item the unknown physical meaning of the boundary scalar field $\phi$:

Although the initial spinor variables attached to the punctures on the space-time corner have a clear interpretation in terms of geometrical flux vectors (representing the triad field) at the classical level and then as ($\SU(2)$-covariant) creation and annihilation operators of quanta of area at the quantum level, the scalar field defined as the real part of the continuum limit of those spinors do not yet have an enlightening physical meaning. On the one hand, it is not clear to which boundary field of general relativity the scalar field $\phi^{A}$ should correspond to. On the other hand, its behaviour under $\SU(2)$ transformations and its action with a specific coupling between its two components should be analysed in detail in order. to understand the defining properties of this boundary field.

\item the apparently asymmetric role of the boundary lapse and boundary shift:

We have started from an action for a discrete set of configurations with a postulated lapse variable $\cN$ enforcing a Hamiltonian constraint and, after taking the continuum limit of the 2d boundary and performing an inverse Legendre transform,  ended up with a covariant action for a boundary (scalar) field living on the time-like boundary of space-time with a 2+1-dimensional metric decomposed in terms of lapse, shift and the 2d corner metric. So this is a paradoxical situation. At the end of the day, the lapse and shift seems to play similar roles as Lagrange multipliers for the space-time boundary diffeomorphim constraints, although in the beginning there was no notion of shift vector. This is due to the a priori discreteness of the quantum spatial geometry versus the assumed continuous nature of time in the loop quantum gravity framework.

We see three (non-excluding and possibly intertwined) ways forward.
First, there could be a (hidden) symmetry of the spinor action for finite $N$ under some specific exchange of area quanta between punctures\footnotemark{}, which would become the invariance under spatial boundary diffeomorphisms in the continuum limit and whose Noether charges should carry a discrete version of the boundary charges for general relativity. One should nevertheless keep in mind that the boundary symmetry depend on the chosen boundary conditions in the continuum theory, and that this begs the question of which type of boundary conditions do boundary flux excitations correspond to.
\footnotetext{
It was hinted in \cite{Girelli:2005ii,Freidel:2009ck,Freidel:2010tt,Livine:2013gna} that the boundary diffeomorphisms in the discrete loop quantum gravity setting should be related to the $\U(N)$ transformations generated by the $\la z_{k}|z_{l}\ra$ observables. It looks more realistic that they might be realised as field-dependent $\U(N)$ transformations (i.e. whose transformation parameters depend on the spinors themselves) (in a similar way that diffeomorphisms can be written as field-dependent translations, e.g. \cite{Freidel:2019ofr}), or that they need to be defined, one level higher, as fusion operators changing the number of punctures $N$ (as when trying to define a discrete equivalent of the Virasoro operators on tensor networks).
}
This leads us to the second point:
the role of boundary conditions. More precisely, should the lapse and shift be dynamical fields on the boundary? Or should they be held fixed on the boundary? In the latter case, we should not work with a Hamiltonian constraint for the spinor dynamics but directly with a Hamiltonian, i.e. we do not require that the Hamiltonian vanishes but it could have a non-zero value, and we do not consider that it generates gauge transformations but more simply boundary symmetries (see \cite{Donnelly:2016auv} for a discussion of gauge versus symmetry). In fact, it is a more general question: which component(s) of the boundary 2+1-d metric should be fixed and which component(s) should be dynamical? There are several apparent possibilities, playing around the lapse, shift, 2d corner metric, the time-like extrinsic curvature, the spatial extrinsic curvature,... This is deeply related to the study of edge modes in general relativity \cite{Freidel:2020xyx,Freidel:2020svx}.
The third path is a more drastic change: if space acquires a discrete nature at the quantum level, then time could/should do. For instance, that's what happens in 't Hooft's polygonal quantization of 2+1-d gravity \cite{tHooft:1993qop,tHooft:1993jwb}, and that's what is postulated in causal dynamical triangulations \cite{Loll:2019rdj}. And that's what is worked out in spinfoam models for a loop quantum gravity path integral (see e.g. \cite{Livine:2010zx,Perez:2012wv} for reviews): space and time intervals become discrete geometrical objects\footnotemark{} (represented as 2d cells).
\footnotetext{
One should distinguish the two different notions (and thereby sources) of discreteness: working with discrete elements of geometry (cells of various dimensions) which we glue together to make the space-time and providing those cells with geometric operators and observables with discrete spectrum. For instance, various area operators for 2d cells in loop quantum gravity can be defined  either with a continuous or discrete spectrum.
This is similar to the discreteness of particles versus the continuity of fields and the discreteness versus the continuity of their momenta.
}
Then the whole 2+1-dimensional boundary need to be described in discrete terms and not only as a discrete 2D spatial boundary evolving in continuous time. This was investigated in the context of the Ponzano-Regge topological state-sum for 2+1-d gravity with time-like boundaries defined at the quantum level in terms of a 1+1-d network structure \cite{Dittrich:2017hnl,Dittrich:2018xuk,Goeller:2019zpz}. However it is then not obvious to come back to a Hamiltonian formulation, which needs to be derived a posteriori from a notion of boundary transfer matrix similarly to what is done with quantum integrable spin systems.

\end{itemize}

{\bigskip}

Despite these gaps in our understanding, the main lesson to draw here is that the polynomial spinor Hamiltonian  for $N$ punctures on the 2d spatial boundary in loop quantum gravity can be written as a $\sigma$-model for a massive two-component scalar field living on the 2+1-d time-like boundary of space-time, with an action roughly of the type $S[\phi]=\int \sqrt{h}\big{[}h^{\mu\nu}K_{AB}\pp_{\mu}\phi^{A}\pp_{\nu}\phi^{B}-m^{2}\phi^{2}\big{]}$. This new correspondence  allows to translate the couplings constants of the spinor Hamiltonian into the components of the 2+1 boundary metric. This boundary metric plays the role of a background geometry in which the flux excitations evolve, just as the spinor couplings dictate the dynamics of the quanta of area living on the punctures in the quantum regime.

In other words, we seek to describe a boundary theory, in which flux excitations evolve over a vacuum state on the space-time boundary. This boundary dynamics  effectively depend on the chosen vacuum state, which translates into particular values for the coupling constants determining the boundary Hamiltonian. In the (na\"ive) continuum limit that we have investigated, these coupling constants translate into a background metric on the 2+1-dimensional time-like boundary. This naturally leads to the conjecture of a correspondence (at the leading polynomial order) between the vaccuum state on loop quantum gravity's corner of space-time (on which the spin network punctures - a.k.a. the area quanta- evolve) and the 2+1 boundary metric.
This conjectured correspondence deserves  further study, in order to determine how far it can be pushed and to which extent it can be made explicit and/or exact.

There remains the (deeper) question of whether the  boundary metric, and thus the coupling constants of the spinor dynamics, should become dynamical, for instance by adding a gravitational term governing the evolution of the boundary metric and thus of the coupling constant. This is the question usually summarised  as ``what boundary conditions should we choose?'' This is deeply intertwined with the definition and fate of gravitational edge modes and with the possibility (or impossibility) of representing (a suitable discrete version of) the boundary charges of general relativity on discrete quantum geometry states.

\section{Lorentz Connection on the Boundary}
\label{sec:sl2c}

In view of the unclear physical meaning of the boundary scalar field derived from the spinors, we would like to investigate an alternative reformulation of the boundary spinor dynamics as a theory of a $\SL(2,\C)$ connection, which could offer further possibilities of comparison with gauge theory reformulations of general relativity boundaries.

This section is thus dedicated to showing how the boundary data defined by the spinors living at the punctures can be described in terms of discrete $\SL(2,\C)$ connections. These connections are required to be flat up to a stabiliser group. The special case of flat connections are also provided with a geometrical interpretation. This hints towards a formulation of boundary theories for loop quantum gravity on space-time corners in terms of $\SL(2,\C)$ gauge theories.

\subsection{$\SL(2,\C)$-holonomies between spinors}
\label{subsec:spinorholo}

Consider two spinors $z_{i}$ and $z_{j}$ living on two punctures. Two arbitrary spinors can not be related by a $\SU(2)$ transformation, since they are not constrain to have equal norm. So mapping a quantum of area onto another can not made solely by a $\SU(2)$ holonomy. This would simply change the (normal) direction of the elementary surface but can not change the size of the quantum of area.
To change the norm of the spinor requires using enlarging the possible transformations by allowing for dilatations, for instance by moving up from $\SU(2)$  to $\SL(2,\C)$ group elements. Indeed, two spinors can be mapped onto each other by $\SL(2,\C)$ group elements. From the perspective of a bulk-to-boundary coarse-graining as discussed in \cite{Livine:2013gna,Livine:2019cvi}, such a Lorentz transformation would account for both the non-trivial $\SU(2)$ transport within the bulk from one boundary puncture to another, but also for the changes of spin occurring at every bulk vertex, as illustrated on fig.\ref{fig:throughthebulk}, leading overall to the propagation from one boundary puncture to another through a $\SL(2,\C)$ holonomy.
\begin{figure}[!htb]
%
%
%
%
%
%
%
\includegraphics[height=40mm]{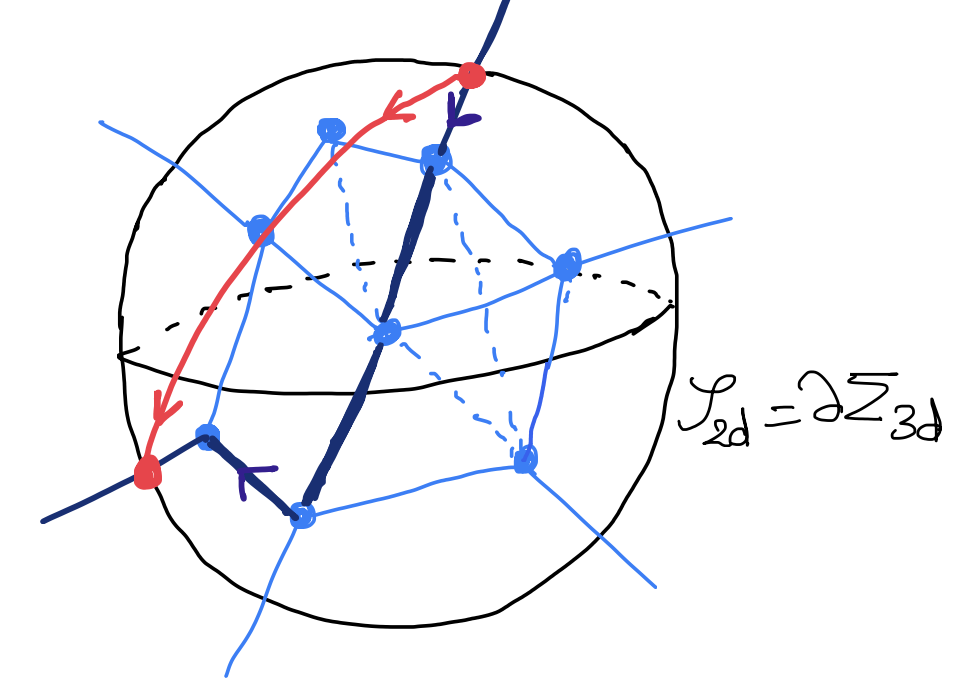}
	\caption{Transport between punctures on the boundary (in red) mapping a quantum of area onto another versus transport along a curve (in bold purple) diving into the bulk following the spin network graph (in blue) spanning the 3d geometry: the boundary transport between punctures, formalized as $\SL(2,\C)$ group elements, can be considered as a coarse-graining of the bulk geometry, projected onto the boundary surface.}
	\label{fig:throughthebulk}
\end{figure}

More precisely, an arbitrary $\SL(2,\C)$ group element is defined by 6 real parameters, while a spinor is determined by 4 real parameter. So the Lorentz holonomy $G_{ij}\in\SL(2,\C)$ relating the two spinors  $z_{i}$ and $z_{j}$ on the boundary is not uniquely determined by those spinors. It is determined up to a stabilizer group element of the initial spinor $z_{i}$, or equivalently of the target spinor $z_{j}$.

Let us choose a reference spinor, the up complex vector $|\uparrow\ra=(1,0)$.
Its stabilizer group consists in the set $\cT\subset\SL(2,\C)$ of upper triangular matrices with trivial diagonal:
\be
|\uparrow\ra=\mat{c}{1 \\ 0}
\,,\qquad
G\,|\uparrow\ra
\,=\,
|\uparrow\ra
\quad\Longleftrightarrow\quad
G
\,=\,
\mat{cc}{1 & \mu \\ 0 &1}
\in\cT
\quad\textrm{with}\,\,
\mu\in\C
\,.
\ee
%
%
Then  $\SL(2,\C)$ group elements mapping the reference spinor to an arbitrary spinor is simply given by the Iwasawa decomposition.
Indeed a $\SL(2,\C)$ group element can be uniquely decomposed as the product of an upper triangular matrix and a $\SU(2)$ group element\footnotemark{}:
\be
G=g \Delta_{\lambda} t_{\mu}\,\qquad\textrm{with}\quad
g\in\SU(2)\,,\quad
\Delta_{\lambda}=\mat{cc}{\lambda & 0 \\ 0 & \lambda^{-1}}\in\cD\,,\quad
t_{\mu}=\mat{cc}{1 & \mu \\ 0 &1}\in\cT\,,
\ee
where $\lambda\in\R_{+}$ defines a dilatation and $\mu\in\C$ defines a translation.
\footnotetext{
One could also assume $\lambda\in\R$, in which case $g$ belongs to $\SU(2)/\Z_{2}\sim \SO(3)$. Or even assume that $\lambda\in\C$, in which case $g\in\SU(2)/\U(1)$ where we remove from $\SU(2)$ the rotations generated by $\sigma_{z}$ and included them in the dilatations.}
This decomposition describes the space of spinors $\C^{2}$ as a section of the coset $\SL(2,\C)/\cT$. Indeed all the $\SL(2,\C)$ transformations mapping the reference spinor to a given spinor $z\in\C^{2}$ are given by a group element $\Lambda_{z}\in\SU(2)\times\cD$ uniquely determined by $z$ times an arbitrary translation $t\in\cT$:
\be
z=\Lambda_{z}t\,|\uparrow\ra
\,,
\qquad\textrm{with}\quad
\Lambda_{z}=\mat{cc}{z^{0} & \f{-\bz^{1}}{\la z|z\ra}\\z^{1} & \f{\bz^{0}}{\la z|z\ra}}
\quad\textrm{and}\quad
t\in\cT\,,
\ee
where the section $\Lambda_{z}$ more precisely decomposes into a dilatation to adjust the norm composed with a $\SU(2)$ rotation:
\be
z
=\Lambda_{z}\,|\uparrow\ra
=g_{\hat{z}}\Delta_{\lambda}\,|\uparrow\ra
\,,\quad
\lambda=\sqrt{\la z|z\ra}
\,,\quad
\hat{z}=\f{z}{\sqrt{\la z|z\ra}}
\,,\quad
g_{\hat{z}}
\,=\,
\f{|{z}\ra\la\uparrow|+|{z}][\uparrow|}{\sqrt{\la z|z\ra}}
\,,\quad 
g_{\hat{z}}\,|\uparrow\ra=\hat{z}
\,.
\ee
Since the stabilizer group of the reference spinor is $\cG_{\uparrow}=\cT$, the stabilizer group for a non-vanishing spinor $z$ is thus obtained by conjugation as $\cG_{z}=\Lambda_{z}\cT\Lambda_{z}^{-1}$.

\medskip

Now we would like to trade the information  of the $N$ spinors $z_{i}$ living on the boundary punctures with the transport information between punctures given $\SL(2,\C)$ holonomies $G_{ij}\in\SL(2,\C)$ such that $G_{ij}\,|z_{i}\ra=\,|z_{j}\ra$. These $\SL(2,\C)$ group elements define a discrete Lorentz connection  on the boundary.

These holonomies clearly can not be arbitrary. Indeed, if we go around a loop from punctures to punctures on the boundary to come back to the initial puncture, the overall holonomy bring map the initial spinor to itself and must therefore lay in its stabilizer. This means that the discrete Lorentz connection is almost flat, in the sense that it must effectively project down to a $\cT$ connection.
More precisely, looking at a loop of punctures $i_{1}\rightarrow i_{2}\rightarrow..\rightarrow i_{n}\rightarrow i_{1}$, the holonomy around the loop is not constrained to the identity but lays in the stabilizer of the initial spinor:
\be
G=G_{i_{n}i_{1}}..G_{i_{1}i_{2}}
\,,\quad
G\,|z_{i_{1}}\ra=\,|z_{i_{1}}\ra
\,,\quad
G\in\cG_{z_{i_{1}}}=\Lambda_{z_{i_{1}}}\cT\Lambda_{z_{i_{1}}}^{-1}
\,.
\ee
In more details, each Lorentz holonomy can be decomposed in $\Lambda$ and triangular matrices:
\be
G_{ij}=\Lambda_{z_{j}}t_{ij}(\Lambda_{z_{i}})^{-1}
\quad\textrm{with}\quad t_{ij}\in\cT\,,\qquad
G_{i_{n}i_{1}}..G_{i_{1}i_{2}}=\Lambda_{z_{i_{1}}}\,(t_{i_{n}i_{1}}..t_{i_{1}i_{2}})\,\Lambda_{z_{i_{1}}}^{-1}
\,.
\ee
The fact that a $\SL(2,\C)$ group element $G$ belongs to one of those stabilizer groups is equivalent to requiring that its trace is equal to 2:
\be
\exists z\in\C^2\,,\,\,G\in\cG_{z}\quad\Longleftrightarrow\quad
\tr\, G =2\,.
\ee
Indeed, if $G$ has its trace equal to 2, and if it belongs to $\SU(2)$ or  if it is a dilatation, then it is necessarily the identity group element.
More technically, $\tr\,G=2$, combined with $\det G=1$, implies that the determinant $\det(G-\id)$ vanishes, which means that $G$ is a triangular matrix with trivial diagonal up to a change of orthonormal basis.

This means that we can replace the boundary data of $N$ spinors $z_{i}$ by the data of a discretized $\SL(2,\C)$ connection, defined by group elements $G_{ij}$ between punctures, with a holonomy constraint around each loop. That holonomy constraint requires that the trace of the holonomy around each loop is equal to 2 and is equivalent to requiring that all those holonomies around loops live in a spinor stabilizer. This can be set as a mathematical proposition:
%
%
\begin{prop}
We call a discrete Lorentz connection on the boundary a collection of $N(N-1)/2$  group elements $G_{ij}\in\SL(2,\C)$ (linking punctures), with the orientation convention that $G_{ji}=G_{ij}^{-1}$. We introduce a holonomy constraint for every  cycle of puncture on the boundary:
\be
\forall \cC \,\textrm{ cycle between boundary punctures }\,,\quad
\tr\bigg{[}\overleftarrow{\prod_{\ell\in\cC}} G_{\ell}\bigg{]}=2
\,.
\ee
This constraint ensures that every Lorentz holonomy around a closed loop is conjugated to a unique group element in $\cT$ (also referred to as a translation).
Then there is a one-to-one correspondence between constrained discrete Lorentz connections $G_{ij}$ and collections of $N$ spinors $z_{i}\in\C^{2}$ (up to a global sign) times discrete $\cT$-connections, i.e. collection of $N(N-1)/2$  group elements $t_{ij}\in\cT$, such that the spinors are transported by the $\SL(2,\C)$ holonomies, $G_{ij}|z_{i}\ra=|z_{j}\ra$, and $G_{ij}=\Lambda_{z_{j}}t_{ij}\Lambda_{z_{i}}^{-1}$.
\end{prop}

\noindent
The sign ambiguity corresponds to a global flip of  the signs, $z_{i}\rightarrow -z_{i}$, which doesn't change the $\SL(2,\C)$ holonomies.

\begin{proof}

From the definitions given above, if we start from $N$ spinors $z_{i}$, we define the group elements $\Lambda_{z_{i}}$ and combine them with the discrete $\cT$-connection $t_{ij}\in\cT$ to define the $\SL(2,\C)$ holonomies between punctures as $G_{ij}=\Lambda_{z_{j}}t_{ij}\Lambda_{z_{i}}^{-1}$. The $\Lambda_{z}$'s were defined to automatically imply  that $G_{ij}|z_{i}\ra=|z_{j}\ra$. Moreover the (ordered) product of the $G_{ij}$ around cycles between punctures automatically gives group elements conjugated to a translation, thus with trace equal to 2.

What remains is to prove the reverse. Let us start with a collection of $\SL(2,\C)$ group elements $G_{ij}$ satisfying the holonomy trace constraint around every cycle on the boundary. We need to reconstruct the spinors and the translations. Let us choose a root puncture $i_{0}$ and consider a cycle $\cC$ starting and finishing at $i_{0}$ (for example, a triangle $i_{0}\rightarrow i_{1}\rightarrow i_{2}\rightarrow i_{0}$). The holonomy $G_{\cC}$ around that cycle has a trace equal to 2 and thus is conjugated to a translation $t_{\cC}$:
\be
G_{\cC}=G t_{\cC} G^{-1}\,,
\ee
for some group element $G\in\SL(2,\C)$. We Iwasawa decompose this group element, $G=g \Delta t$, so that $G_{\cC}=(g\Delta) t_{\cC} (g\Delta)^{-1}$. We define the spinor at the puncture $i_{0}$ as the image of the up spinor by $g \Delta$:
\be
z_{i_{0}}=g\Delta \,|\uparrow\ra\,,
\quad\textrm{which implies that} \quad \Lambda_{z_{i_{0}}}=g\Delta
\,.
\ee
We can then transport this spinor to define the spinors at every other puncture, $z_{i}=G_{i_{0}i}z_{i_{0}}$. We still have to check that this is a consistent definition, i.e. that $z_{j}=G_{ij}z_{i}$. This is equivalent to checking that $z_{i_{0}}$ is stabilized by the holonomy around the triangle  $i_{0}\rightarrow i\rightarrow j\rightarrow i_{0}$:
\be
z_{j}=G_{ij}z_{i} \Leftrightarrow  G_{i_{0}j}z_{i_{0}}=G_{ij}G_{i_{0}i}z_{i_{0}}
\Leftrightarrow  z_{i_{0}}=G_{i_{0}j}^{-1}G_{ij}G_{i_{0}i}z_{i_{0}}
\,.
\ee
This means that we would have obtained the same spinor $z_{i_{0}}$ if we had used the cycle $i_{0}\rightarrow i\rightarrow j\rightarrow i_{0}$ instead of $\cC$.

To show this, it is enough to consider two cycles $\cC_{1}$ and $\cC_{2}$, with $\SL(2,\C)$ holonomies $G_{1}$ and $G_{2}$, starting at $i_{0}$ and finishing at $i_{0}$, and prove that the spinor stabilized by $G_{1}$ is also stabilized by $G_{2}$. As explained above, since $G_{1}$ is conjugated to a translation, we can write:
\be
G_{1}=\Lambda_{\om_{1}} t_{\mu_{1}}\Lambda_{\om_{1}}^{-1}=\id+\mu_{1}|\om_{1}\ra[\om_{1}|\,,
\ee
for a spinor $\om_{1}\in\C^{2}$, and the same for $G_{2}$. The holonomy for the concatenated cycle $(\cC_{1};\cC_{2})$ then reads:
\be
G_{2}G_{1}
=
\id+\mu_{1}|\om_{1}\ra[\om_{1}|+\mu_{2}|\om_{2}\ra[\om_{2}|
+
\mu_{1}\mu_{2}[\om_{2}|\om_{1}\ra\,|\om_{2}\ra[\om_{1}|\,,\qquad
\tr G_{2} G_{1}=2-\mu_{1}\mu_{2}[\om_{2}|\om_{1}\ra^{2}\,.
\ee
The holonomy constraint around the concatenated cycle $(\cC_{1};\cC_{2})$ thus implies that the holomorphic scalar product between the two spinors vanishes, $[\om_{2}|\om_{1}\ra=0$, which in turn implies that $\om_{2}$ is proportional to $\om_{1}$. In particular, $\om_{2}$ is also stabilized by $G_{1}$ and vice-versa $\om_{1}$ is also stabilized by $G_{2}$.

\end{proof}

We would like to point out a variation about the reconstruction of $\SL(2,\C)$ holonomies from the boundary spinor data. It is traditional in standard works in loop quantum gravity to split the spinor $z\in\C^{2}$ in its flux vector $\vX\in\R^{3}$ (carrying information about the embedding of the surface patch within the 3d space) and its phase (interpreted as the twist angle, carrying information about the 2+1 embedding of the surface patch in the time direction). Putting the phase aside and focusing on the flux vector, one can introduce $\SL(2,\C)$ group elements mapping the flux vector $\vX_{i}$ at one boundary puncture onto the flux vector $\vX_{j}$ at another puncture. As we show in the appendix \ref{app:vector-sl2C}, this boost action has a $\SU(1,1)$ stabilizer group, thereby realizing a 3+3 splitting of $\SL(2,\C)$ instead of the 4+2 splitting used above when acting on spinors. The holonomy condition is then looser. Instead of enforcing that the trace of the $\SL(2,\C)$ around boundary cycles is necessarily 2, it can now be an arbitrary real number strictly larger than 2. We do not pursue in this direction and focused instead of the spinor variables.

\medskip

Here, we will not study in details the possible symplectic structures that one can endow the space of discrete (constrained) Lorentz connections with and the question of their matching with the canonical Poisson bracket on the spinors. There is not a straightforward obvious answer. On the one hand, there are (well-known) ambiguities on the definition of symplectic structures on discrete Lorentz connections and $\cT$ connections, and on the other hand, it is not clear if the $\Lambda_{z}$'s are the best choice of section and what is supposed to be the brackets of the spinor $z$ with the translation parameter $\mu$ of the upper triangular matrix.
On top of these ambiguities, there is also the question of the precise role of the holonomy constraints. For instance, should we define the  symplectic structure before or after imposing the holonomy constraints? Do the holonomy constraints form a set of first class constraints generating a gauge invariance and leading to a symplectic quotient? Following previous work on the combinatorial quantization of Chern-Simons theory \cite{Buffenoir:2002tx} and the related phase space of 2+1-d loop quantum gravity with a cosmological constant \cite{Dupuis:2013haa,Bonzom:2014wva,Bonzom:2014bua,Dupuis:2014fya,Dupuis:2019yds}, it is tempting to hope that the holonomy constraints would generate a kind of translations, but we postpone such analysis to future investigation.

We will focus instead on the geometrical interpretation of discrete Lorentz connections, on the formulation of a dynamical boundary theory and their possible continuum limit.


\subsection{Non-trivial stabilizer and (relative) locality on the boundary}

In this setting with the correspondence between spinors and  discrete $\SL(2,\C)$-connections, the next natural question is whether we can identify a geometrical or physical meaning to the extra data carried by the Lorentz connections compared to the spinors, i.e. provide an interpretation to the stabilizer group data and the discrete $\cT$-connection.

Since the stabilizer group $\cT$ is a two-dimensional abelian Lie group, isomorphic to $\C$ (provided with the addition), we propose to interpret these group elements as translation on the 2d boundary. More precisely, assuming that the 2d boundary has a spherical topology, we see it as a complex manifold and parametrize it in terms of a complex variable $\zeta$ locating points on the boundary. We propose to interpret the extra data contained in the $\SL(2,\C)$  holonomies on top of the spinors as position coordinates $\zeta_{i}\in\C$ for each puncture.

\medskip

Let us take a step back and reflect on the structure of (boundary) surfaces in loop quantum gravity.
A surface consists in a set of quanta of areas carried by the spin network punctures. These quanta of areas are defined mathematically as spin states resulting from the (canonical) quantization of spinors. Each spin state is geometrically interpreted as an elementary surface patch carrying a (quantized) vector, whose norm gives the quantized area in Planck unit and whose direction is the normal direction to the surface, and a $\U(1)$ phase called the twist angle, which indicates the extrinsic curvature integrated over the surface patch. In this traditional setting for loop quantum gravity, there is absolutely no information on where a surface patch is located with respect to other patches on the overall (boundary) surface.

Indeed, in the background independent framework of loop quantum gravity, it is the spin network state - its underlying graph and the algebraic data dressing it- which defines the 3d space (quantized) geometry. Reconstructing the overall 3d geometry of a (large) region of the spin network is not straightforward, it is a non-local reconstruction and recognising two graph nodes as close or far is a hard question\footnotemark.
\footnotetext{There is actually no rigorous theorem proving that this reconstruction is systematically possible and unique. The reader will find a very interesting discussion of possible non-locality effects resulting from the background independence of spin network states in \cite{Markopoulou:2007ha}.}
In fact, considering a region with a 3-ball topology (and its boundary with a 2-sphere topology), it is the bulk spin network state data that determines the 3d geometry and thus the notion of locality on its 2d boundary. If we focus (too much) on the algebraic data induced on the boundary - the spin states- and discard all the other bulk information, we lose the possibility to localise points on the boundary.

This issue is circumvented for 3d regions with a  single node. This corresponds to a single quantum of volume, defined by the intertwiner state carried by the node. The intertwiner is interpreted as a quantized convex polyhedron. Indeed, considering the  3-vectors induced by the spin states around the node,  there exists a unique convex polyhedron such that these are the normal vectors of the polyhedron's faces. This is ensured by Minkowski's theorem for convex polytopes: in this algorithm assuming the convexity of the surface, the normal vectors play the double role of determining the direction of the boundary face and determining each face position with respect to the other faces.
A spin network is then interpreted as a discrete 3d geometry resulting from gluing those quantized convex polyhedra together \cite{Freidel:2010aq,Bianchi:2010gc}.
Such a convex polyhedron interpretation is fine and well-suited for an elementary quantum of volume, but does not really make sense for a general surface bounding a possibly large region of the quantum space. Indeed, there is absolutely no reason to assume that the boundary surface is convex (with respect to the bulk 3d geometry).
This is the deep reason -to properly represent the observables for 2d surfaces- why there recently has been several proposals of extending the algebraic data carried by spin networks and augment them with  2d metric data \cite{Freidel:2015gpa,Freidel:2016bxd,Freidel:2018pvm,Freidel:2019ees,Freidel:2019ofr}.

From this perspective, it seems very appealing to be able to naturally incorporate localization data (2d coordinates) for the boundary punctures with the spinors (area excitations) into $\SL(2,\C)$ holonomies. More precisely, we show below that there is an enticing one-to-one correspondence between punctures provided with spinor and position and flat discrete Lorentz connections in the boundary.

\medskip

Assuming that 2d space boundary has a  spherical topology, we view the 2-sphere  as a complex manifold and parametrize it in terms of a complex variable $\zeta$ locating points on the boundary. We provide every punctures on the boundary with position coordinates, given by an extra complex variable $\zeta_{i}$ associated to each puncture. Then we  show below that we can define a unique Lorentz holonomy between two punctures, which transports the extended data defined y the spinor-position pair $(z_{i},\zeta_{i})$ from one puncture onto another.

Drawing inspiration from previous works on $\SU(2)$ twisted geometries and $\SL(2,\C)$ spin networks \cite{Freidel:2010bw,Livine:2011vk,Speziale:2012nu}, we embed the complex position $\zeta_{i}$  as a spinor $w_{i}\in\C^{2}$ up to complex rescaling:
\be
w=\mat{c}{w^{0}\\ w^{1}}\sim\lambda w \,,
\,\,\forall\lambda\in\C
\qquad\longrightarrow\quad
\textrm{equivalence class}
\quad
[w]=[\lambda w] 
\quad\textrm{defined by ratio}\,\,
\zeta=\f{w^{0}}{w^{1}}
\,.
\ee
Then we define the unique $\SL(2,\C)$ group element between two punctures which maps $(z_{i},w_{i})$ onto $(z_{j},w_{j})$:
\be
G_{ij}=
\f{|z_{j}\ra[w_{i}|-|w_{j}\ra[z_{i}|}{[w_{i}|z_{i}\ra}
\,\in\SL(2,C)
\,,\qquad
G_{ij}\,|z_{i}\ra=|z_{j}\ra
\,,\quad
G_{ij}\,|w_{i}\ra=|w_{j}\ra\,.
\ee
This Lorentz holonomy transports as wanted the spinor and position from one puncture to another, $G_{ij}\triangleright (z_{i},\zeta_{i})=(z_{j},\zeta_{j})$.
Its definition is possible and valid if and only if the (holomorphic) scalar product between spinors remains constant $[w_{i}|z_{i}\ra=[w_{j}|z_{j}\ra$.
This gives a unique prescription for the $\SL(2,\C)$ holonomy in terms of the spinors $(z_{i},z_{j})$ and positions $(\zeta_{i},\zeta_{j})$. Indeed this necessary condition fixes the spinors $w$'s in terms of the complex coordinates $\zeta$'s, up to an arbitrary overall scale fixed for the whole boundary network.
%
%
Let's indeed choose a global value  $[w_{i}|z_{i}\ra=\sigma\in\C$ for all the punctures $i=1..N$.
Then if we are given the spinors $z_{i}\in\C^{2}$ and the complex positions $\zeta_{i}$, we can reconstruct the spinors $w_{i}$:
\be
w_{i}=\mat{c}{\lambda_{i}\zeta_{i}\\ \lambda_{i}}
\,,\qquad
\sigma=[w_{i}|z_{i}\ra
=\lambda_{i}\,\big{(}
\zeta_{i}z^{1}_{i}-z^{0}_{i}
\big{)}
\quad\Rightarrow\quad
\lambda_{i}=\f{\sigma}{\zeta_{i}z^{1}_{i}-z^{0}_{i}}
\,\in\C
\,.
\ee
This leads to unique $\SL(2,\C)$ group elements, which are actually independent of the specific value chosen for $\sigma$:
\be
G_{ij}
=
\f1\sigma\mat{cc}
{\lambda_{j}z^{1}_{i}\zeta_{j}-\lambda_{i}z^{0}_{j} & \lambda_{i}\zeta_{i}z^{0}_{j}-\lambda_{j}z^{0}_{i}\zeta_{j} \\
\lambda_{j}z^{1}_{i}-\lambda_{i}z^{1}_{j} & \lambda_{i} \zeta_{i}z^{1}_{j} -\lambda_{j}z^{0}_{i}}
\,.
\ee
We can recast this $\SL(2,\C)$ holonomy in the factorized form involving the $\cT$-holonomy $t_{ij}$ and the section group elements $\Lambda_{z_{i}}$,
\be
G_{ij}=\Lambda_{z_{j}}\mat{cc}{1 & \mu_{ij}\\ 0 & 1}\Lambda_{z_{i}}^{-1}
=
\f{|z_{j}\ra \la z_{i}|}{\la z_{i}|z_{i}\ra}+\f{|z_{j}][ z_{i}|}{\la z_{j}|z_{j}\ra}+\mu_{ij}|z_{j}\ra [ z_{i}|
\,.
\ee
Matching this with the formulas above gives the expression of the $\cT$-holonomy coefficient $\mu_{ij}$ between punctures in terms of the complex coordinates $\zeta_{i}$ of the punctures: 
\be
\mu_{ij}=
\mu_{j}-\mu_{i}
\,,\qquad\textrm{with}\quad
\mu_{i}
=
\f1{\la z_{i}|z_{i}\ra}
\left(\f{\zeta_{i}\bz^{0}_{i}+\bz^{1}_{i}}{z^{0}_{i}-\zeta_{i}z^{1}_{i}}\right)
\,.
\ee
%
Since the group elements $G_{ij}$ are uniquely determined by each pair of source and target spinor-position $(z_{i},\zeta_{i},z_{j},\zeta_{j})$,  this construction defines a flat discrete $\SL(2,\C)$ connection on the boundary surface, $G_{i_{n}i_{1}}..G_{i_{1}i_{2}}=\id$ for every cycle of punctures on the boundary $i_{1}\rightarrow i_{2}\rightarrow i_{3}\rightarrow\dots\rightarrow i_{n}\rightarrow i_{1}$. This leads to the following proposition:

\begin{prop}

Considering the boundary data defined by a collection of $N$ punctures each dressed with a spinor $z_{i}\in\C^{2}$ and a complex coordinate $\zeta_{i}$ with $i$ labelling the punctures and running  from 1 to $N$, this defines a unique discrete Lorentz connection with $\SL(2,\C)$ group elements $G_{ij}$   transporting the boundary data from the puncture $i$ to the puncture $j$:
\be
G_{ij}z_{i}=z_{j}\,,\qquad
G_{ij}\mat{c}{\zeta_{i}\\1}\propto\mat{c}{\zeta_{j}\\1}\,,
\ee
where $G_{ij}$ acts by multiplication by the corresponding 2$\times$2 matrix and the proportional relation is up to an arbitrary (non-vanishing) complex number.%
This discrete connection is necessarily flat, i.e. the $\SL(2,\C)$ holonomy around any cycle $\cC$ on the boundary is trivial:
\be
\overleftarrow{\prod_{\ell\in\cC}} G_{\ell}=\id\,.
\ee
Reciprocally, a discrete flat Lorentz connection defines boundary data $(z_{i},\zeta_{i})_{i=1..N}$ up to a global $\SL(2,\C)$ transformation.

\end{prop}
\begin{proof}
 
We have already constructed the $\SL(2,\C)$ holonomies from the spinor-position data $(z_{i},\zeta_{i})_{i=1..N}$. We simply have to describe the reverse procedure. Considering a discrete flat Lorentz connection $G_{ij}\in\SL(2,;\C)$, we start from one arbitrarily chosen puncture, $i_{0}$ and we choose arbitrary data $(z_{i_{0}},\zeta_{i_{0}})\in\C^{2}\times\C$. Then we define all the remaining spinor-position variables by transport from $i_{0}$ as $(z_{i},\zeta_{i})=G_{i_{0}i}\triangleright(z_{i_{0}},\zeta_{i_{0}})$. The flatness condition for the connection, $G_{ij}=G_{i_{0}j}G_{i_{0}i}^{-1}$ around a triangle, ensures that $(z_{j},\zeta_{j})=G_{ij}\triangleright(z_{i},\zeta_{i})$ for all pairs of  punctures. The global $\SL(2,\C)$ freedom corresponds to the choice of the initial data  $(z_{i_{0}},\zeta_{i_{0}})$.

\end{proof}
\begin{figure}[!htb]
%
%
%
%
%
%
%
\includegraphics[height=55mm]{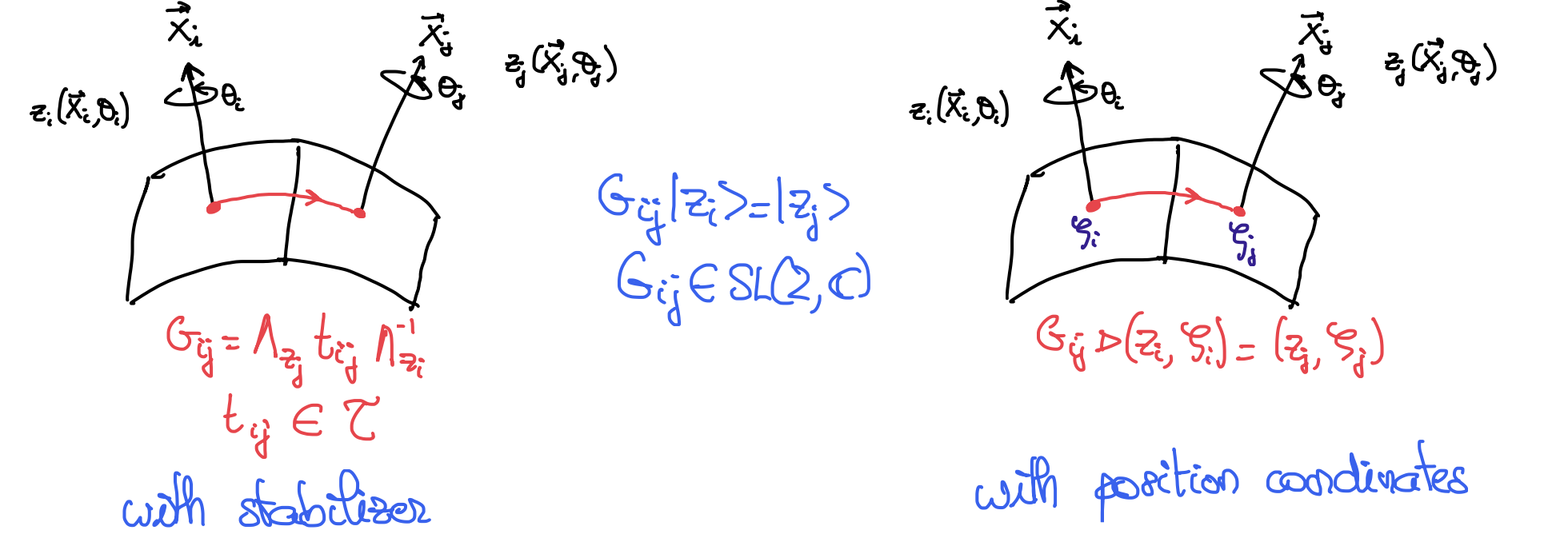}
	\caption{The spinor $z$, living at each puncture, encodes the (semi-classical) information about the quantum of area carried by the puncture. It contains both the flux vector $\vX$ giving the normal vector to the surface patch and the twist angle $\theta$ representing the measure of extrinsic curvature to the spatial slice.  We define $\SL(2,\C)$ boundary holonomies $G_{ij}$ that transport the spinor information from puncture to puncture, such that $G_{ij}|z_{i}\ra=|z_{j}\ra$. Either we consider the spinors as the whole boundary data and the $\SL(2,\C)$ holonomies are determined by the spinors up to stabilizer group elements given by triangular matrices $t_{ij}\in\cT$ (or the left hand side). Or we supplement the spinors $z_{i}$ with complex variables $\zeta_{i}$ indicating the position of the puncture on the 2d boundary and then the $\SL(2,\C)$ holonomies are entirely determined by the spinors  and complex coordinates (on the right hand side). In the first case, the holonomies of the discrete $\SL(2,\C)$ connection  live in the stabilizer subgroup $\cT$. While, in the latter case, the resulting  discrete $\SL(2,\C)$ connection is exactly flat.}
	\label{fig:boundaryposition}
\end{figure}

The geometric picture is that the complex variables $\zeta_{i}$, parametrizing the spinor stabilizer and interpreted as 2d  coordinates, play the role of discrete 2d metric data on the boundary. From this point of view, they describe the intrinsic geometry of the boundary. On the other hand, the spinors $z_{i}$ define the normal 3d vectors to the surface and thus describe the extrinsic geometry of the boundary (within the 3d slice). Together, they describe the whole discrete embedded geometry of the boundary surface within the spatial slice, as illustrated on fig.\ref{fig:boundaryposition}.

If we compare the spinor dynamics ansatz introduced in the previous section \ref{subsec:dynamics}, the complex coordinates $\zeta_{i}$ are meant to play a similar role than the coupling constants of the spinor Hamiltonian, which were shown to give the boundary metric. In this formulation in terms of $\SL(2,\C)$ connections, the intrinsic geometry, encoded by the $\zeta_{i}$'s, are put on the same footing than the extrinsic geometry, encoded by the $z_{i}$'s. Both will naturally acquire dynamics. There is obvious reason to keep the $\zeta_{i}$'s or $z_{i}$'s fixed while the other evolve and fluctuate, except if we make a special choice of boundary conditions. For instance, we could keep the intrinsic 2d metric fixed on the boundary surface and let the embedding data evolve, which would correspond to fixing the complex coordinates $\zeta_{i}$'s and the spinor norms $\la z_{i}|z_{i}\ra$.

\bigskip

Since we have identified field configurations corresponding to flat Lorentz connections and understood them in terms of boundary geometry, it is natural to think about curvature defects and seek a geometrical interpretation for $\SL(2,\C)$ connections with non-trivial curvature. 
For instance, in the more general setting where we focus solely on the spinors, as above in section \ref{subsec:spinorholo}, we do not require the $\SL(2,\C)$ connection to be flat, but allow for a non-trivial stabilizer. This means that, if we start at a puncture dressed with spinor $z$ and complex position $\zeta$ and go around a cycle on the boundary back to that initial puncture, the non-trivial holonomy $G\in\cG_{z}=\Lambda_{z}\cT\Lambda_{z}^{-1}$ will shift the 2d coordinate from $\zeta$ to a new position $\tilde{\zeta}$. Or in short, if we go around a loop, the position of the puncture has changed.
This is reminiscent of the framework of relative locality \cite{AmelinoCamelia:2011bm,AmelinoCamelia:2011pe,Freidel:2013rra}, where non-trivial connection and torsion in phase space lead to a relativity of the position of an event with respect to the observer and its history.
That would provide a geometrical interpretation for translational curvature defects of the Lorentz connection. It is tempting to try to interpret $\SU(2)$ curvature defects as magnetic excitations of the Ashtekar-Barbero connection along the directions tangent to the boundary (and not transversally) and dilatation curvature defects as some conformal excitations, but we leave this for future analysis.


\subsection{$\SL(2,\C)$  boundary theory}

Now that we have reformulated the flux excitations (or area quanta) of loop quantum gravity on space-time corners in terms of  (discrete) flat Lorentz connections, it is natural to think about its continuum limit as a theory of a Lorentz connection field on the 2+1-d time-like boundary of space-time.
If we want to discuss the specifics of the boundary theory and dynamics, we need to first address two questions:
\begin{itemize} 
\item As it is natural with a  Lorentz connection, should we impose gauge invariance under local Lorentz transformations?
\item As boundary flux excitations correspond to flat Lorentz connections, should we focus on solely on flat connections or more broadly consider theories of arbitrary Lorentz connections whose equations of motion in vaccuum (i.e. without defects or sources) nevertheless impose flatness?
\end{itemize}

Starting with the issue of enforcing or not local gauge invariance, one needs to keep in mind that there is a single flat Lorentz connection on the 2-sphere up to gauge transformations. So, if one focuses on purely flux excitations (thus only area quanta and not dual defects on the boundary) and thus considers only flat connections, imposing local $\SL(2,\C)$ gauge invariance on the boundary renders the boundary theory trivial and empty. Indeed local $\SL(2,\C)$ transformations at a puncture means that one can arbitrarily shift the spinor $|z\ra\in\C^{2}$ living at the puncture. Since the spinor indicates both the size of the area quantum and its (normal) direction, enforcing $\SL(2,\C)$ gauge invariance would mean that all quanta of area are physically equivalent. This would amount to requiring the invariance of the boundary theory under both 2d diffeomorphisms and conformal transformations. Although this seems perfectly acceptable, reasonable and even preferable, that goes against the traditional view of boundaries in loop quantum gravity where the spins carried by each boundary puncture is thought of as a physical observable.

\medskip

So let us start with the case where we do not require a local gauge invariance under boundary $\SL(2,\C)$ transformations. We will consider the case of gauge-invariant dynamics later in this section.
So the boundary variables are the $\SL(2,\C)$ holonomies $G_{ij}$ between punctures. In fact, these determine the spinors (and positions) only up to a global $\SL(2,\C)$ transformation, which correspond to the choice of spinor (and position) at a chosen root puncture $i_{0}$. So we would like to define a theory, its action principle and Hamiltonian, in terms of  group elements $G_{ij}\in\SL(2,\C)$ between punctures and one spinor $z_{i_{0}}\in\C$. From this spinor at the root puncture, one immediately reconstructs all the spinors by transporting it to the other punctures, assuming that $z_{i}=G_{i_{0}i}z_{i_{0}}$. Then we can immediately reformulate the spinor action \eqref{actionN} and Hamiltonian \eqref{boundaryH} defined in the previous section by expressing the scalar product between spinors in terms oif the $\SL(2,\C)$ holonomies:
\be
\la z_{i} |z_{j}\ra=\la z_{i_{0}}| G_{i_{0}i}^{\dagger}G_{i_{0}j}|z_{i_{0}}\ra\,,\qquad
[ z_{i} |z_{j}\ra=[ z_{i_{0}}| G_{i_{0}i}^{-1}G_{i_{0}j}|z_{i_{0}}\ra=[ z_{i_{0}}| G_{ij}|z_{i_{0}}\ra
\,,
\ee
where we have used the flatness condition of the discrete Lorentz connection.
This leads to an action written in its canonical form which reads:
\be
S[\{G_{kl}\}, z_{i_{0}}]
=
\int\rd t\, \bigg{[}
-i\la z_{i_{0}}| \sum_{k}G_{i_{0}k}^{\dagger}G_{i_{0}k}|\rd_{t}z_{i_{0}}\ra
-i\sum_{k}\la z_{i_{0}}| G_{i_{0}k}^{\dagger}\rd_{t}G_{i_{0}k}|z_{i_{0}}\ra
-\cN\cH
\bigg{]}\,,
\ee
\beq
\textrm{with}\quad
\cH&=&\sum_{n}\beta_{n}\sum_{k}\la z_{i_{0}}| G_{i_{0}k}^{\dagger}G_{i_{0}k}|z_{i_{0}}\ra^{2n}
+\gamma_{2}\sum_{k,l}C_{kl}\la z_{i_{0}}| G_{i_{0}k}^{\dagger}G_{i_{0}l}|z_{i_{0}}\ra\nn\\
&&+\tilde{\gamma}_{2}\sum_{k,l}D_{kl}\la z_{i_{0}}| G_{kl}|z_{i_{0}}\ra
+\bar{D}_{kl}\la z_{i_{0}}| G_{kl}|z_{i_{0}}\ra
+\dots
\nn
\eeq
This is the exact equivalent of the boundary spinor action ansatz \eqref{boundaryH}  proposed earlier.
Here we have implicitly assumed the flatness of the Lorentz connection, through the triangular relation between  $\SL(2,\C)$ group elements $G_{kl}=G_{i_{0}k}^{-1}G_{i_{0}l}$. The flatness condition allows to transport all the spinors back to the root puncture. We could put both the flatness condition and the use of a root spinor aside and generalize the Hamiltonian above to define an action principle in terms of solely the discrete $\SL(2,\C)$ connection. Indeed, forgetting about $z_{i_{0}}$, the natural proposal is to replace the projection on $|z_{i_{0}}\ra\la z_{i_{0}}|$ by the trace of the group elements, thus writing the Hamiltonian as a linear combination of terms $\tr\, G_{kl}^{\dagger}G_{kl}$ and $\tr G_{kl}$ and their powers:
\be
S_{\SL(2,\C)}[\{G_{kl}\}]
=
\int\rd t\, \bigg{[}
-i\sum_{k,l}\tr \,G_{kl}^{\dagger}\rd_{t}G_{kl}
-\cN\cH_{\SL(2,\C)}
\bigg{]}\,,
\ee
\be
\textrm{with}\quad
\cH_{\SL(2,\C)}=\sum_{n}\beta_{n}\sum_{k,l}(\tr \,G_{kl}^{\dagger}G_{kl})^{2n}
+\gamma_{2}\sum_{k,l}C_{kl} \tr \,G_{kl}^{\dagger}G_{kl}\nn\\
+\tilde{\gamma}_{2}\sum_{k,l}D_{kl}\tr\, G_{kl}
+\dots
\nn
\ee
The next step would be to analyze the physics predicted by such dynamics.
This would require 1. understand the solution to the equations of motion; 2. check the stability or not of flat connections; 3. work out the continuum limit of this discrete ansatz as a field theory.

\medskip

This conclude the proposal for a boundary theory non gauge-invariant under local $\SL(2,\C)$ transformations. We now turn to the possibility of defining gauge invariant boundary dynamics. Let us first point out that the bulk (loop) quantum gravity is invariant under two kinds of gauge transformations: $\SU(2)$ gauge transformations and (space-time) diffeomorphisms. Since $\SL(2,\C)$ is (much) larger than $\SU(2)$, requiring the $\SL(2,\C)$ gauge invariance of the boundary is a non-trivial extension of the $\SU(2)$ gauge transformations generated by the Gauss law of the Ashtekar-Barbero connection: it could involve boundary diffeomorphisms or other types of transformations (e.g. conformal transformations,\dots). It should thus correspond to a very specific class of boundary conditions with enhanced symmetry.

As we underlined earlier, since there is a unique flat $\SL(2,\C)$ connection on the boundary 2-sphere up to $\SL(2,\C)$ gauge transformations, enforcing the gauge invariance of the boundary theory under local $\SL(2,\C)$ transformations leads to a trivial theory with a single physical state. This amounts to considering all possible area quanta on the spatial boundary as physically equivalent, whatever their direction and size. This would become interesting only if we venture away from the flatness condition of the $\SL(2,\C)$ connection. This follows a perfect natural logic: 
trading the spinors representing boundary flux excitations for  Lorentz connections, thus identifying the flux excitations as  flat Lorentz connections, naturally leads to contemplating the meaning of non-flat connections, which  should represent different types of boundary excitations.
In the continuum limit, we are therefore looking for a gauge field theory of a connection, whose equations of motion in vaccuum (i.e. without source or defect) amount to the flatness of the connection. Natural candidates are coset $\SL(2,\C)$ Chern-Simons theories (see e.g. \cite{Isidro:1991fp}) or a $\SL(2,\C)$ Yang-Mills theory on the 2+1-d boundary. What needs to be understood is 1. the (quantum) boundary conditions that they generate on the 2+1-d time-like boundary of space-time; 2. if those boundary conditions are compatible with the bulk dynamics of (loop) quantum gravity (i.e. generated by the bulk Hamiltonian constraints for some choice of space-time foliation). 

One should nevertheless keep in mind that we are not yet searching for a unique boundary theory, but more for classes of boundary theories. They should correspond to classes of boundary conditions for the bulk fields. They will likely admit for a non-trivial renormalization group flow corresponding to the changes of (quantum) boundary conditions under dilatations and deformations of the boundary.

\subsection{Recovering local $\SU(2)$ gauge invariance on the boundary: magnetic excitations}
\label{sec:gauge}

We have discussed the possibility of defining boundary dynamics invariant or not under $\SL(2,\C)$ local gauge transformations. However, the natural set of gauge transformations in the loop quantum gravity consists in $\SU(2)$ transformations. These act as local 3d rotations on the flux. The question is then if it makes sense to impose a $\SU(2)$  gauge invariance on the boundary (and not a $\SL(2,\C)$  gauge invariance anymore).

The $\SU(2)$ action on the boundary spinors is the straightforward  multiplication by 2$\times$2 matrices:
\be
|z\ra\in\C^{2}\,\mapsto h\,|z\ra\in\C^{2}\,,\qquad h\in\SU(2).
\ee
The boundary spinor dynamics  \eqref{actionN}-\eqref{boundaryH} that we have studied up to now has already been assumed to be (gauge) invariant under global $\SU(2)$ transformations, leading to an expansion of the Hamiltonian in terms of scalar products between the spinors. Now we would like to upgrade this invariance to local $\SU(2)$ transformations. This is naturally achieved by introducing extra degrees of freedom on the boundary\footnotemark: a (discrete) $\SU(2)$ connection defining the transport of the flux excitations -the spinors- on the boundary. 
\footnotetext{
Another -na\"ive- way to impose invariance under local $\SU(2)$ transformations is to restrict to the $\SU(2)$-invariant part of each flux excitation. This means keeping only the norm of each spinor $\la z_{k}|z_{k}\ra$, or equivalently the spin $j_{k}$ of the corresponding quanta of area, discarding any information about the direction and phase of the spinor $z_{k}$, or equivalently about the magnetic moment and phase of the spin state living in $\cV_{j_{k}}$. A boundary state would simply be labelled by the spins, without further data (i.e. no spin state or intertwiner). The dynamics would then couple those spins $j_{k}$ together and let them evolve. From the point of view of the 2d geometry of the boundary surface, this corresponds to keeping only the density factor (determinant of the induced 2d metric) as boundary data and discarding the rest of the 2d metric and all the extrinsic data about e embedding of the 2d boundary surface in the 3d space.  Although this could make sense mathematically, we do not see any physical reason for such drastic reduction of the boundary modes.
}

We thus consider $\SU(2)$ group elements $g_{kl}$ associated to (oriented) pairs of punctures $(k,l)$. While local $\SU(2)$ transformations acts locally on the spinors at each puncture, they act at both source and target of the $\SU(2)$ holonomies:
\be
z_{k}\mapsto h_{k}z_{k}\,,\qquad
g_{kl}\mapsto h_{l}g_{kl}h_{k}^{-1}
\,.
\ee
Thus the quadratic $\SU(2)$-invariant combination of spinors are the transported scalar products, $\la z_{l}|g_{kl}|z_{k}\ra$. We can define the boundary theory as before in section \ref{subsec:dynamics}, with a polynomial Hamiltonian in $\la z_{l}|g_{kl}|z_{k}\ra$ and $[ z_{l}|g_{kl}|z_{k}\ra$. The action principle would now depend on both the spinors $z_{k}$ and the $\SU(2)$ holonomies $g_{kl}$ as independent variables. In the continuum limit, these $\SU(2)$ holonomies would become a $\SU(2)$ connection field, defining a $\SU(2)$ covariant derivative on the time-like boundary of space-time.

It is very tempting to interpret this boundary $\SU(2)$ connection as the (pull-back of the) Ashtekar-Barbero connection on the space-time corner. This means that we do not consider only the Ashtekar-Barbero connection in the direction transversal to the boundary surface -carried by the spin network edges puncturing the surface- but also its components tangential to the boundary surface -thus to be carried by spin network links running along the surface\footnotemark. These are interpreted as {\it magnetic} excitations of the boundary supplementing the {\it electric} excitations defined by the flux living at the boundary punctures (see e.g. \cite{Freidel:2019ofr} for a discussion of the various edge modes of loop quantum gravity). This type of extended spin network structure, with links running transversally and tangentially to surfaces, was already proposed in e.g. \cite{Charles:2016xzi,Delcamp:2018efi,Freidel:2019ees}, and it would be interesting to investigate further the corresponding possible  boundary dynamics.
\footnotetext{
The special role fo tangential spin network links was already discussed in early loop quantum gravity work, especially for their non-trivial contribution to the area spectrum and to black hole entropy computations \cite{Frittelli:1996cj,Rovelli:1998gg}.
}

The last point we would to discuss is the relation between the newly introduced $\SU(2)$ connection $g_{kl}$ on the boundary and the $\SL(2,\C)$ connection $G_{kl}$ that we have discussed up to now. The essential difference is that the $\SU(2)$ connection is a field independent from the spinors while the $\SL(2,\C)$ connection depends on the spinors and is actually a reformulation of the spinors. Thus, the $\SL(2,\C)$ transport between punctures exactly maps the spinors onto each other, $G_{kl}|z_{k}\ra=|z_{l}\ra$, while the $\SU(2)$ transport between punctures simply allows to write down locally $\SU(2)$-invariant scalar products $\la z_{l}|g_{kl}|z_{k}\ra$.
Nevertheless, depending on the precise boundary action and Hamiltonian, the two connections could be related on-shell. Indeed, if as an example the Hamiltonian consisted in $-\sum_{k,l}\la z_{l}|g_{kl}|z_{k}\ra$, its minimal value\footnotemark{} for fixed spinors $z_{k}$ would be obtained if and only if the $\SU(2)$ group elements exactly transport the spinors, i.e. $|z_{l}\ra=g_{kl}|z_{k}\ra$. Then the $\SU(2)$ connection could be identified as the $\SU(2)$ part of the $\SL(2,\C)$ connection.
\footnotetext{
This Hamiltonian $\sum_{k,l}\la z_{l}|g_{kl}z_{k}\ra$ is actually real if we assume the natural orientation condition that $g_{lk}=g_{kl}^{-1}$. Indeed, this means that $\la z_{k}|g_{lk}z_{l}\ra=\la z_{k}|g_{kl}^{\dagger}|z_{l}\ra=\overline{\la z_{l}|g_{kl}|z_{k}\ra}$.
}
We postpone a more in-depth analysis of the possible coupled dynamics of boundary spinors and boundary $\SU(2)$ connection, i.e. of both electric and magnetic excitations on the boundary- to future investigation.

\section*{Outlook \& Conclusion}

We have investigated the basic structure of boundary theories for loop quantum gravity, meaning the dynamics of the degrees of freedom induced by the fluctuations and evolution of the bulk geometry on the 2+1-dimensional time-like boundary of space-time. This 2+1-dimensional boundary is seen as the time evolution of the 2d space boundary surface. As spin network states  span and create the quantum geometry of the 3d space, the spin network links puncture the boundary surface. The spin states carried by those boundary punctures define quanta of area of the surface. These are also referred to as flux excitations or ``electric'' excitations. They are the basic boundary data of loop quantum gravity. 

We showed that the dynamics of those quanta of area on the boundary surface can be mathematically formulated in terms of complex 2-vectors -spinors- attached to the  boundary punctures. Considering a Hamiltonian polynomial in those excitations, the lowest order Hamiltonian (with up to quartic terms) corresponds to a  Bose-Hubbard model (with two species). Such condensed matter models are known for their phase transitions, which could lead to interesting new physics and predictions for loop quantum gravity phenomenology. Furthermore, we showed that the coupling constants of the quadratic terms of the Hamiltonian can be understood in the continuum limit (as the number of punctures is sent to infinity, thereby creating a lattice structure on the boundary surface) as the components of a  background metric on the 2+1-d time-like boundary. We consider this result as the hint of a deeper correspondence between the coupling constants of the boundary Hamiltonian for area quanta and (quantum) states of the 2+1-d boundary metric (or 2+1-d dressed metric in semi-classical terms).

We further showed for to reformulate the boundary spinor data in terms of discrete flat $\SL(2,\C)$ connections. Curvature of this boundary $\SL(2,\C)$ connection should be understood as novel boundary excitations for loop quantum gravity, which we speculate to correspond to the ``translational'' and ``magnetic'' edge modes envisioned in \cite{Freidel:2019ofr}. This opens the door to formulating boundary theories for loop quantum gravity as $\SL(2,\C)$ gauge theories (such as e.g. coset Chern-Simones theories).

\medskip

In the present work, our starting point was the formalism of loop quantum gravity and its spin network states for the quantum geometry. We have tried to clarify what could/would be a boundary theory in this framework and one should now analyze the physics of those models and compare them with the (many) recent works on boundary theories and edge modes in classical general relativity. In order to tackle this next step, there is first the broader question of what is meant by  {\it a boundary theory} for (quantum) gravity. We see three levels of sophistication:
\begin{enumerate}
\setcounter{enumi}{-1}
\item {\it boundary terms  to the bulk action:}

One usually needs to add a boundary term to a field theory action principle in order to ensure the differentiability of the action with respect to field variations and cleanly derive the equations of motion for the bulk field while respecting chosen boundary conditions. For instance, the Gibbons-Hawking-York term in the metric formulation of general relativity, defined as the boundary integral of the extrinsic curvature, ensures the differentiability of the corrected Einstein-Hilbert action (space-time integral of the scalar curvature) when fixing the induced metric on the boundary. This typically further ensures the gauge invariance of the overall action - bulk action plus boundary term- under gauge transformations consistent with the boundary conditions. This can be done consistently in a covariant phase space approach.

\item {\it dynamical boundary conditions and edge modes:}

The covariant phase space goes further and allows to define dynamical boundary conditions. Indeed, in a canonical setting, we distinguish the canonical boundary - the 3d space for general relativity - and the time-like boundary - the 2+1-d boundary, which describes the evolution in time of the 2d boundary of the 3d space. We usually refer to the 2d spatial  boundary surface as the corner (of space-time). Instead of fixing boundary conditions on the whole 2+1-d time-like boundary, one seeks to identify boundary field degrees of freedom living on the corner - typically, in general relativity's metric formulation, the 2d metric on the corner plus extra fields (e.g. a radial expansion scalar and a shear vector)- and let them evolve in time. The time-like boundary conditions are thus generated as the time evolution of initial boundary conditions on the corner. One usually requires that this boundary evolution satisfies a suitably extended gauge-invariance of the original bulk theory.
The obvious question is whether the time evolution of the boundary data on the corner can be derived as a Hamiltonian dynamics defined on its own, i.e. solely in terms of the boundary data without any reference to the bulk fields\footnotemark{}.
For such classes of boundary conditions, this edge mode dynamics defines the boundary theory.
\footnotetext{
This seems to be the natural (simplest) setting for (quasi-local) holography (from a canonical/hamiltonian viewpoint). In the case that one can identify classes of boundary conditions such that the boundary decouples from the bulk, in the sense that the edge mode dynamics can be defined entirely in terms of boundary fields, there is a holographic duality when the resulting boundary theory carries a representation of the conserved charges of the bulk dynamics: then the boundary and bulk theories have the same symmetry algebra and one can seek a correspondence between the observables of the two theories. A weaker form is when the boundary and bulk theories are Morita-equivalent, i.e. do not necessarily have the same symmetry algebra but their algebra of observables have isomorphic categories of representations.
}

\item {\it effective boundary theory:}

A more intricate setting is when the boundary dynamics does not decouple from the bulk theory and thus can not be defined on its own. One can nevertheless hope to integrate over the bulk fields and dynamics and extract an effective theory for the edge modes, i.e. an effective boundary theory. As a coarse-graining of the bulk theory\footnotemark{},
\footnotetext{Thus, there is no a priori reason for this effective boundary theory to be unitary. One should instead account for possible dissipation in the bulk or in other words, flux across the boundary.
}
this procedure can be understood as studying the renormalisation group flow for the 2+1-d boundary theory\footnotemark{}. This totally reverse the original logic of the field theory: instead of fixing boundary conditions and analysing the resulting bulk field and (quantum) fluctuations, one integrate over the bulk degrees of freedom and focus on the induced dynamics of the boundary conditions.
\footnotetext{
Holography seems to be achieved at a fixed point of this renormalisation flow, when the effective theory boundary is given by the original boundary theory -the part of the dynamics proper to the boundary without the interaction terms with the bulk- up to possible shifts of the coupling constants (i.e. without any new interaction terms).
}

\end{enumerate}

The present work, discussing the possible dynamics of quanta of area on the 2+1-d time-like boundary of space-time, is naturally set at the level 1 of this hierarchy. The next step in this direction would thus be to compare the boundary theory ansatz which we introduced with edge mode hamiltonians derived from general relativity in the case that the boundary conditions admit a proper hamiltonian formulation (i.e. if the edge mode dynamics can be defined without referring to the bulk degrees of freedom).
From this perspective, the natural questions to face are:

\begin{itemize}
\item how do the presently introduced templates for the dynamics of boundary area quanta compare the edge mode phase space and dynamical boundary conditions in general relativity?
\item do those models of boundary dynamics carry a representation of the boundary charge algebra (i.e. of the symmetry algebra) of general relativity or of a suitably deformed version to be implemented i quantum gravity?
\item do they define renormalizable quantum field theories in the continuum limit? does this renormalization flow correctly represent deformation and rescaling of the boundary surface? What are the  fixed points (as conformal field theories) of this flow?
\end{itemize}
Beside connecting the present work to the study of edge modes in gravity theories, we would like to conclude on the fact that it also opens new perspectives. Indeed the boundary dynamics models we introduced can  naturally be  written as condensed matter models with non-trivial phase diagrams, and thus with possible interesting phase transitions, which would enrich the phenomenology of (loop) quantum gravity.

\section*{Acknowledgement}

I am very grateful for the regular discussions on edge modes in quantum field theory and holography in quantum gravity with Laurent Freidel, Marc Geiller and Christophe Goeller.

\appendix
%

\section{$\SL(2,\C)$-holonomy between 3-vectors}
\label{app:vector-sl2C}


A variation of our new framework for LQG boundaries in terms of $\SL(2,\C)$ holonomies is to define a more equal splitting of the $\SL(2,\C)$ data, as a pair of 3-vectors instead of a 4-dimensional spinor and a 2-dimensional complex number. This can be thought as natural. In the context of twisted geometries, a spinor $z\in\C^2$ encodes a 3-vector $\vec{X}\in\R^3$ plus a phase $\theta$. The 3-vector is interpreted as the normal vector to the boundary surface, while the twist angle $\theta$ is  canonically conjugate to the 3-vector norm and usually interpreted as encoding the extrinsic curvature across the boundary surface. A splitting in terms of pairs of 3-vectors would amount to bundling the twist angle together with the complex position.

More precisely, to each puncture is associated a 3-vector $\vX_{i}$, which can be defined mathematically as a spinor $z_{i}$ up to a phase:
\be
\vX=\la z|\vsigma|z\ra
\,,\quad
|\vX|=\la z|z\ra
\,,\quad
|z\ra\la z|=\f12\,\left(|\vX|\,\id\,+\,\vX\cdot\vsigma\right)
\,,\quad
X=\vX\cdot\vsigma=2|z\ra\la z|-\la z|z\ra\id
\,,
\ee
where the $\sigma$'s are the Pauli matrices.
Considering two 3-vectors $\vX_{i}$ and $\vX_{j}$, or equivalently the corresponding two traceless Hermitian matrices $X_{i}$ and $X_{j}$, we can map one onto the other with a Lorentz transformation $G_{ij}\in\SL(2,\C)$:
\be
G_{ij}X_{i}G_{ij}^\dagger=X_{j}\,.
\ee
Let us start with the unit 3-vector along the $z$-axis and see how to boost it to an arbitrary 3-vector $X$, i.e. what $\SL(2,\C)$ realizes $G\sigma_{z}G^\dagger=X$. We consider a spinor $z\in\C^2$ such that $\vX=\la z|\vsigma|z\ra$ and identify a group element $G$ mapping the reference spinor to $z$:
\be
G\,|\uparrow\ra=z
\quad\Rightarrow\quad
G\sigma_{z}G^\dagger=X\,.
\ee
The stabilizer group of $\sigma_{z}$ under the action by conjugation obviously contains all the transformations generated by $\sigma_{z}$, that is all the complex dilatations $\Delta_{\lambda}$ with $\lambda\in\C$, but is actually larger. It is the $\SU(1,1)$ subgroup of $\SL(2,\C)$:
\be
G\sigma_{z}G^\dagger=\sigma_{z}
\quad\Longleftrightarrow\quad
G=\mat{cc}{a & b \\ \bar{b} & \bar{a}}
\qquad\textrm{with}\quad
|a|^2-|b|^2=1\,.
\ee
So the stabilizer group for the reference 3-vector $X=\sigma_{z}$ is $S^0=\SU(1,1)$ and we can rotate it to get the stabilizer group for an arbitrary vector $X\in\R^3$ obtained as $S_{X}=g_{\hat{X}}S^0g_{\hat{X}}^{-1}$, where $g_{\hat{X}}$ is a $\SU(2)$ group element rotating the $z$-axis unit vector to the direction of the 3-vector $\hat{X}=\vX/|\vX|$.

Similarly to what we have done in the spinorial case, we can now trade the set of 3-vectors $\vX_{i}$ living at the punctures on the boundary for a network of $\SL(2,\C)$ holonomies $G_{ij}$ between punctures on the boundary. Requiring that the holonomies map the 3-vectors from one puncture to another, $G_{ij}\act X_{i}=G_{ij}X_{i}G_{ij}^\dagger=X_{j}$ means that all $\SL(2,\C)$-holonomies $G$ around loops on the boundary have to live in a $\SU(1,1)$ stabilizer subgroup. Algebraically, this translates to the condition $\tr\,G\in\R$, $\tr\,G >2$.


\bibliographystyle{bib-style}
\bibliography{LQG}

\end{document}